\documentclass[twoside,11pt]{article}
\usepackage{jair}
\usepackage{times}  
\usepackage{helvet}  
\usepackage{courier}  
\usepackage{url}  
\usepackage{graphicx}  
\usepackage{paralist}
\usepackage[inline]{enumitem}
\usepackage[utf8]{inputenc}
\usepackage{etex}
\usepackage{amssymb,amsthm,amsmath}
\usepackage{graphicx}
\usepackage[usenames, dvipsnames]{color} 
\usepackage{stmaryrd}
\usepackage[noend,ruled]{algorithm2e} 
\usepackage[font={small}]{subfig}
\usepackage{tikz}
\usepackage{pgfplots,filecontents}
\usepackage{adjustbox}
\usepackage{graphics}
\usepackage{booktabs}
\usepackage{multirow}
\usepackage{enumitem}  
\usepackage{todonotes}
\usepackage{url}
\usepackage{algorithm2e}

\newcommand{\Imc}{\ensuremath{\mathcal{I}}}

\newcommand{\I}{\ensuremath{\mathcal{I}}\xspace}
\newcommand{\C}{\ensuremath{\mathcal{C}}\xspace}
\newcommand{\T}{\ensuremath{\mathcal{T}}\xspace}
\newcommand{\A}{\ensuremath{\mathcal{A}}\xspace}
\newcommand{\K}{\ensuremath{\mathcal{K}}\xspace}
\newcommand{\R}{\ensuremath{\mathcal{R}}\xspace}
\newcommand{\ISA}{\ensuremath{\sqsubseteq}}
\newcommand{\ent}{\ensuremath{\vDash}}

\newcommand{\subsume}{\ensuremath{\sqsubseteq}}

\newcommand{\term}{\ensuremath{\emph{term}}}

\newcommand{\dllite}{\textit{DL-Lite}\xspace}

\newcommand{\dlliteR}{\ensuremath{\textit{DL-Lite}_\R}\xspace}

\newcommand{\dlliter}{\dlLiteR}
\newcommand{\dlLiteR}{\dlliteR}

\newcommand{\newDL}{\mbox{$\dllite^{\mathcal{HR}}$}\xspace}
\newcommand{\newDLnr}{\mbox{$\dllite_\mathsf{non\mbox{-}rec}^{\mathcal{HR}}$}\xspace}

\newcommand{\vars}{\mathit{vars}}

\newcommand{\conc}[1]{\mathrm{#1}}
\newcommand{\role}[1]{\mathrm{#1}}
\newcommand{\ind}[1]{\mathsf{#1}}

\newcommand{\nc}{{\sf N_C}}
\newcommand{\nr}{{\sf N_R}}

\newcommand{\nind}{{\sf N_I}}

\newcommand{\nrsimple}{{\sf N_{R_{s}}}}

\newtheorem{definition}{Definition}
\newtheorem{theorem}{Theorem}

\newtheorem{lemma}{Lemma}
\newtheorem{example}{Example}
\newtheorem{claim}{Claim}
\newtheorem{proposition}{Proposition}

\newcommand{\myspe}{{\leadsto_{\T}}\,}

\newcommand\xqed[1]{%
  \leavevmode\unskip\penalty9999 \hbox{}\nobreak\hfill
  \quad\hbox{#1}}
\newcommand\demo{\xqed{$\triangle$}}

\newcommand{\Smod}[2]{\ensuremath{\mathcal{E}_{#1 , #2}}\xspace}
\newcommand{\ord}{\mathit{ord}} 
\newcommand{\bfA}{\mathbf{A}\xspace} 

\newtheorem{innercustomthm}{Lemma}
\newenvironment{lemma2}[1]
{\renewcommand\theinnercustomthm{#1}\innercustomthm}
{\endinnercustomthm}

\newtheorem{innercustomthmT}{Theorem}
\newenvironment{theorem2}[1]
{\renewcommand\theinnercustomthmT{#1}\innercustomthmT}
{\endinnercustomthmT}

\newtheorem{innerproposition}{Proposition}
\newenvironment{proposition2}[1]
{\renewcommand\theinnerproposition{#1}\innerproposition}
{\endinnerproposition}

\newcommand{\myspeS}{{\leadsto_{\T}^s}\,}
\newcommand{\mygen}{{\leadsto_{\T}^g}\,}
\newcommand{\gsArrow}[2]{\ensuremath{\leadsto_{#1}^{#2}}}
\newcommand{\gsArrowT}[2]{\ensuremath{{{\leadsto_{#1}^{#2}}^*}}}

\usetikzlibrary{arrows}
\usetikzlibrary{decorations.pathreplacing}

\tikzstyle{every node} = [font=\scriptsize, text centered]
\tikzstyle{every label} = [inner sep=0pt, font=\scriptsize, text centered]

\tikzstyle{aboxNode} = [rectangle,rounded
corners,fill=blue!10,draw=black] 

\tikzstyle{anon} = [rectangle,rounded
corners,fill=yellow!15,draw=black] 

\tikzstyle{varnode} = [node/.style={rectangle,rounded corners,draw,align=center}] 
\tikzstyle{role} = [->,shorten <=2pt,shorten >=2pt]

\tikzstyle{match} = [blue,densely dotted,->,shorten <=-1pt,shorten >=-1pt]
\tikzstyle{matchR} = [BrickRed,densely dotted,->,shorten <=-1pt,shorten >=-1pt]

\tikzstyle{nosepn} = [outer sep=1pt, inner sep=1pt]

\tikzstyle{querymatch} = [PineGreen,dashed,->]

\begin{document}
%
\title{
Relaxing and Restraining Queries for OBDA 
}

\author{Medina Andre\c{s}el \and Yazm\'in Ib\'a\~nez-Garc\'ia \and Magdalena
  Ortiz 
\and Mantas \v{S}imkus \\
\{andresel,ibanez,ortiz\}@kr.tuwien.ac.at $\mid$ simkus@dbai.tuwien.ac.at \\
Faculty of Informatics, TU Wien, Austria \\ 
}

\maketitle
\begin{abstract}
In ontology-based data access (OBDA), ontologies have been successfully
employed for querying possibly unstructured and incomplete data. In this
paper, we advocate using ontologies not only to formulate queries and compute
their answers, but also for modifying queries by relaxing or restraining
them, so that they can retrieve either more or less answers over a given
dataset. Towards this goal, we first illustrate that some domain knowledge
that could be naturally leveraged in OBDA can be expressed using complex role inclusions (CRI). 
Queries over ontologies with CRI are not first-order (FO) rewritable in general.  
We propose an extension of DL-Lite with CRI, and show that  conjunctive queries over ontologies in this extension 
are FO rewritable.    
Our main contribution is a set of rules to relax and
restrain conjunctive queries (CQs). Firstly, we define rules that use the
ontology to produce CQs that are relaxations/restrictions over any dataset. Secondly, we introduce a set of data-driven rules, that leverage
patterns in the current dataset, to obtain more fine-grained relaxations and
restrictions.
\end{abstract}


\maketitle \acks{This research was funded by FWF Projects P30360 and W1255-N23.} 

\section{Introduction}
\emph{Ontology based data access (OBDA)} is one of the most
successful use cases of description logic (DL) ontologies. 
The core idea in  OBDA is to use an ontology to provide a conceptual view of a collection of data sources, 
thus abstracting away  from the specific way data is
stored. The role of the ontology in this setting is to describe the domain
of interest at a high level of abstraction. This allows users to formulate queries over the data sources using a familiar controlled  vocabulary. 
Further, knowledge represented in the ontology can be leveraged to retrieve more complete answers.  
For example,  
consider the dataset $\A_e$ in Figure~\ref{fig:eventData}, which includes
information on some cultural events and their locations, and the
 ontology in Figure~\ref{fig:eventOnt} which captures additional information, e.g. the knowledge that  
both concerts and exhibitions are cultural events. 
By posing  the query
\[ q_1(x)\gets\conc{CulturEvent}(x) \]
one can retrieve all cultural events, $\sf ex_1$, $\sf ev_1$, and
$\sf c_1$, regardless  whether they are stored as concerts,
exhibitions, or cultural events of unspecified type. 

In the OBDA paradigm an ontology can be linked to a collection of heterogeneous data sources by defining \emph{mappings} from the data to the vocabulary used in the  ontology~\cite{OBDA}. This allows to integrate e.g., data from relational databases and unstructured datasets. In this framework, an ontology acts as a \emph{mediator} between the query and a set of heterogeneous data sources.  
%
Description logics of the so-called \dllite family have been particularly tailored for 
OBDA~\cite{CalvaneseGLLR07}. One crucial property of  \dllite, is that queries mediated by such ontologies are first-order(FO)-rewritable. 
In a nutshell, this means that evaluating a query $q$ over a dataset $\A$ using knowledge from an ontology  $\T$
can be reduced to evaluate  a query $q_\T$, that incorporates knowledge from $\T$, over $\A$. This amounts to query evaluation in relational databases.   
In our example, a rewriting of $q$ is 
\[ q_\T(x) \gets \conc{CulturEvent}(x) \lor \conc{Exhibition}(x) \lor
\conc{Concert}(x) \]
\begin{figure*}[t]
	\begin{minipage}[t]{0.5\linewidth}
		\includegraphics[width=\textwidth]{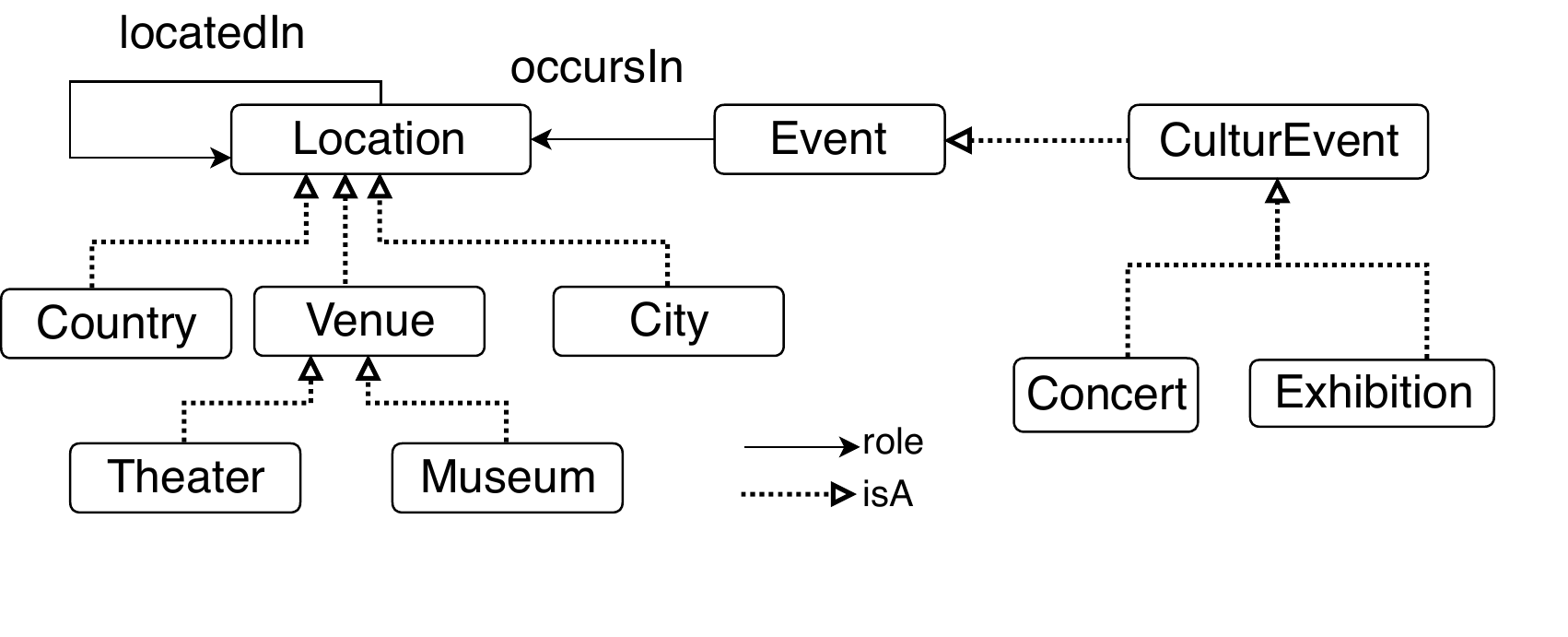}
		\caption{Event ontology $\T_e$.}
		\label{fig:eventOnt}
	\end{minipage}
	\hfill
	\begin{minipage}[t]{0.45\textwidth}
		\begin{small}
			%
			\begin{tikzpicture}
			\node (Vie) [aboxNode, label={[shift={(0,-0.8)}]$\conc{City}$}] at (0,0) {$\sf Vienna$} ;
			\node (Au) [aboxNode, label={[shift={(0,-0.8)}]$\conc{Country}$}] at (2,0) {$\sf Austria$};
			\node (StOper) [aboxNode, label={[shift={(0,-0.8)}]$\conc{Venue}$}] at (-2,0) {$\sf StateOpera$};
			\node (c1) [aboxNode, label={[shift={(0,0.2)}]$\conc{Concert}$}] at (-2,1.6) {$\sf c_1$};
			\node (ex1) [aboxNode, label={[shift={(0,0.2)}]$\conc{Exhibition}$}] at (0,1.6) {$\sf ex_1$} ;
			\node (ev1) [aboxNode, label={[shift={(0,0.2)}]$\conc{CulturEvent}$}] at (2,1.6) {$\sf ev_1$};
			\draw [role] (Vie) -- (Au) node[midway,above=5pt] {$\role{locatedIn}$};
			\draw [role] (StOper) -- (Vie) node[midway,above=5pt, ] {$\role{locatedIn}$};
			\draw [role] (ex1) -- (Vie) node[midway,above=2pt, rotate=90] {$\role{occursIn}$};
			\draw [role] (ev1) -- (Au) node[midway,above=2pt, rotate=90] {$\role{occursIn}$};
			\draw [role] (c1) -- (StOper) node[midway,above=2pt, rotate=90] {$\role{occursIn}$};
			\end{tikzpicture}
		\end{small}
		\caption{Event dataset $\A_e$.}
		\label{fig:eventData}
	\end{minipage}
	\vspace{-0.6cm}
\end{figure*}

We investigate the use of ontologies not only as query mediators, but also for 
\emph{query reformulation}: by modifying queries in order to \emph{relax} them and retrieve more answers, or  \emph{restrain} them and reduce answers. These reformulations can be used to explore a given 
dataset, or to modify queries to fit the information needs of a user.
For example, answers to queries for concerts may be too scarce or nonexistent, then by relaxing the query  
to find all cultural events, one might get more answers.  Conversely, if a query for cultural events produces too many
answers, it is possible to restrict this query to events of a specific type (for instance concerts).


In our example, the query 
$q_c(x) \gets \conc{Concert}(x)$ that specializes  
$q_1(x) \gets \conc{CulturEvent}(x)$ occurs as a disjunct in its rewriting.
A key observation within our approach is  that query
rewriting rules for \dllite (such as the ones from ~\cite{CalvaneseGLLR07}) 
yield query specializations, and that
counterparts of these rules produce query
generalizations. 
However, there are intuitive query answers and query reformulations 
that cannot be produced by these rewriting rules. 
For example,  consider the following  query retrieving concerts occurring in Vienna. 
 \[q_2(x) \leftarrow \conc{Concert}(x) \land \role{occursIn}(x,y) \land y={\sf
   Vienna}\] 
This  query does not return any answers when evaluated over the data from $\A_e$ w.r.t.\,
$\T_e$. However, ${\sf c_1}$ may be considered an answer to this 
query, by following the intuition that if an event occurs in a venue located in a city, then it occurs in that city.
Still, this knowledge cannot be expressed in \dllite.  

To obtain this kind of reformulations, we propose to extend \dllite with  so-called \emph{complex
role inclusions} (CRI). In our example we could add the following: 
		\[\role{occursIn} \cdot \role{locatedIn} \subsume
                \role{occursIn} \]
This axiom captures the intuition above, and would allow to retrieve ${\sf
  c_1}$. 
Moreover, we could also use it to generate some
interesting query reformulations. For instance,  the query
\[q_3(x) \leftarrow \conc{Concert}(x) \land \role{occursIn}(x,y) \land
\role{City}(y).\]
could be specialized to  
 \begin{align*}
q'_3(x) \leftarrow \, & \conc{Concert}(x) \land \role{occursIn}(x,z)  \land \role{locatedIn}(z,y) \land \role{City}(y)
 \end{align*}
which specializes from all concerts known to occur in a city, to only those for which a more
specific location within a city is known.

Unfortunately, adding CRIs to \dllite increases the
worst case data complexity of query answering, which means that queries are no longer  
FO-rewritable. We  propose two extensions of \dllite with CRIs for which queries remain FO-rewritable. 
The first extension imposes  some acyclicity conditions 
between the roles that occur on the right-hand-side of CRIs. 
This extension, however, would not be sufficient to capture our example above,
where $\role{occursIn}$ appears on both sides of the inclusion. 

A more expressive  extension of \dllite allowing recursive role inclusions
can be defined based on the observation that chains of some roles have bounded length. 
In our example, we note that concepts occurring along chains  of the role  $\role{locatedIn}$ can be 
\emph{ordered} in the sense that  $\role{occursIn}$ edges can only 
connect `smaller' locations to `larger' ones: from  venues to cities, from
cities to countries, etc. 
Based on this observation we propose yet another extension of  \dllite allowing recursion along ordered bounded concept chains. 
We then propose reformulation rules for relaxing and specializing queries using ontologies
in this extension. The resulting rules allow to generalize and specialize queries ``moving'' not only along the subclass relation and subrole relations, but also along dimensions defined by the ordered concepts (in our example along the different kinds of locations).

Using ontologies expressed in the proposed extension of \dllite is possible to reformulate queries along dimensions expressed at the intentional level. 
%
However, there are some intuitive reformulations  that cannot be obtained  on the basis of an ontology alone. 
Let us illustrate this in our example. Recall the query $q_2$ asking for concerts occurring in Vienna. 
It could be  specialized, for instance, to
concerts in some venue in Vienna, like the State Opera, or  generalized  to
all concerts in Austria. This can only be done by taking into consideration the dataset at hand (that is the intentional knowledge).

To capture this intuition, we  propose rules considering instances of concepts and relations, 
as well as inclusions between concepts that are not necessarily implied by the TBox, 
but that can be guaranteed to hold in the current dataset.   
Applying the resulting rules to $q_2$ 
produces 
the following reformulations: 
\begin{align*}
q^s_2(x) \leftarrow {} & \conc{Concert}(x) \land \role{occursIn}(x,y) \land y={\sf
   StateOpera}\\
q^g_2(x) \leftarrow {} & \conc{Concert}(x) \land \role{occursIn}(x,y) \land y={\sf
   Austria} 
\end{align*}
Note that these reformulations are not data independent, but instead,
refer to $\A_e$. 

The proposed  ontology and data-driven reformulations can aid users to explore heterogeneous,  
unstructured and incomplete datasets in the same spirit as online
analytical processing (OLAP) supports the exploration of structured data~\cite{citeulike:907825}. 
For that purpose, we illustrate how our extension of 
\dllite can describe \emph{dimensional} knowledge, analogous to the
multi-dimensional data model considered in OLAP. We also exemplify how our rules for relaxing and restraining queries can be applied 
in a way that closely resembles the so-called `rolling up' and `drilling down'  along
dimensions.

\section{Preliminaries} 
We start by introducing the syntax and semantics of \dlliteR~\cite{CalvaneseGLLR06}. 
We assume an alphabet consisting of countable infinite sets $\nc$,$\nr$, $\nind$ 
of \emph{concept}, \emph{role}, and \emph{individual} names, respectively. 
%
%
\dlliteR expressions 
are constructed according to the following grammar: 
\begin{align*}
B  &:= \top \mid \bot \mid A \mid \exists r &    r &:= p \mid p^-, 
\end{align*} 
where $A \in \nc$, $p \in \nr$. Concepts of the form $B$ are called \emph{basic concepts}, and roles of the form $p^-$ are called \emph{inverse roles}.  A \dlliteR \emph{TBox} (or \emph{ontology}) is a finite set of  axioms of the form  
\[B_1 \ISA B_2, \quad  r_1 \ISA r_2, \quad \mathbf{disj}(B_1,B_2), \quad \mathbf{disj}(r_1,r_2).\]  
%
A \dlliteR \emph{ABox} (or \emph{dataset}) is a finite set of assertions of the forms $A(a)$, and  $p(a,b)$, 
with $a,b \in \nind$, $A \in \nc$, and $p\in \nr$. A knowledge base (KB) is a pair $\K= (\T, \A)$.

The semantics is defined as usual in terms of interpretations.  
An \emph{interpretation} $\Imc= (\Delta^\Imc, \cdot^\Imc)$ consists of a non-empty \emph{domain} $\Delta^\Imc$ and an \emph{interpretation function} $\cdot^\Imc$ assigning to every concept name a set $A^\Imc \subseteq \Delta^\Imc$, and to every role name $p$ a binary relation $p^\Imc \subseteq \Delta^\Imc \times \Delta^\Imc$. The interpretation of more complex concepts and roles 
is defined as follows:
\begin{align*}
&\top^\Imc = \Delta^\Imc,   &(\exists r)^\Imc &= \{d \mid \exists d'. (d,d') \in r^\Imc \}, \\
&\bot^\Imc = \emptyset,  &(p^-)^\Imc &= \{(d,d') \mid (d',d) \in p^\Imc \} %
\end{align*}
Further, each individual name in $\nind$ is interpreted as an element  $a^\Imc \in \Delta^\Imc$, such that $a^\Imc = a$ for every $a \in \nind$ 
(i.e., we adopt the standard name assumption).

An interpretation $\Imc$ \emph{satisfies} an axiom of the form $\alpha
\ISA \beta$ if $\alpha^\Imc \subseteq \beta^\Imc$, an axiom of the
form $\mathbf{disj}(\alpha, \beta)$ if $\alpha^\Imc \cap
\beta^\Imc=\emptyset$, an assertion $A(a)$ if $a \in A^\Imc$, and an
assertion $p(a,b)$ if $(a,b) \in p^\Imc$.
Finally, $\Imc$ is a \emph{model} of a KB $\K= (\T, \A)$, denoted
$\Imc \models \K$, if $\Imc$ satisfies every axiom in \T, and every
assertion in \A. An ABox $\A$ is \emph{consistent} with a TBox $\T$ if there
exists a model of the KB $(\T, \A)$.
\begin{example}[Event KB]
	The ontology in Figure \ref{fig:eventOnt} is formalized into the  \dlliteR TBox $\T_e$:
\begin{align*}
\exists\conc{locatedIn} & \subsume \conc{Location} \, &\exists \role{occursIn} & \subsume \conc{Event}  \\
 \exists\conc{locatedIn}^- & \subsume \conc{Location} \, &\conc{CulturEvent} &\subsume \conc{Event}  \\
\exists \role{occursIn}^- & \subsume \conc{Location} \, 	&\conc{Concert} &\subsume \conc{CulturEvent}  \\
\conc{Country} & \subsume \conc{Location}  \, &\conc{Exhibition} &\subsume \conc{CulturEvent} \\
 \conc{Venue} & \subsume \conc{Location} \, &\conc{Theater} & \subsume   \conc{Venue}  \\
 \conc{City} & \subsume \conc{Location} \, &\conc{Museum} & \subsume   \conc{Venue}   
\end{align*}

\noindent The dataset $\A_e$ in Figure \ref{fig:eventData} together with $\T_e$ form a \dlliteR KB, which we  
denote as \emph{event KB} $\K_e$.
\label{ex:eventKB}
\demo
\end{example}

\paragraph{Normal form.}
W.l.o.g., in the rest of this paper we will consider  TBoxes in 
\emph{normal form}. In particular, we assume that all axioms in a TBox
$\T$ have one of the following forms: (i) $A\ISA A'$, (ii) $A\ISA
\exists p$, (iii) $\exists p \ISA A $, (iv) $p\ISA s$, (v) $p\ISA
s^{-}$, (vi) $\mathbf{disj}(A,A')$, and (vii) $\mathbf{disj}(p,p')$, where $A,A' \in \nc$ and $p,p',s \in \nr$. We note that by using (linearly many) fresh symbols, a general TBox can be transformed into a TBox in normal form so that the models are preserved up to the original signature.


\paragraph{Queries.}
We  consider the class of conjunctive queries and unions thereof. 
A \emph{term} is either an individual name or a variable.   
A \emph{conjunctive query} 
is a first order formula with free variables $\vec{x}$ that takes the form 
$\exists \vec{y}. \varphi(\vec{x}, \vec{y})$, with $\varphi$ a conjunction of 
\emph{atoms} of the form $$A(x), \quad r(x,y), \text{ and } \quad t = t',$$ where
$A \in  \nc$, $r \in \nr$,  and  
$t,t'$ range over terms. The set of terms occurring in a query $q$ is denoted $\term(q)$.
The free variables on a query are called the \emph{answer variables}. We  use
the notation $q(\vec{x})$ to make explicit reference to the 
answer variables of $q$. The \emph{arity of $q(\vec{x})$} is defined as the length of $\vec{x}$, denoted $|\vec{x}|$. Queries of arity 0 are called \emph{Boolean}. We  sometimes omit the existential variables and use 
$$q(\vec{x}) \leftarrow \varphi(\vec{x}, \vec{y}),$$  to denote a query
$q(\vec{x})= \exists \vec{y}. \varphi(\vec{x},\vec{y})$. Further, when operating on
queries,  it will be convenient to identify a CQ $q(\vec{x}) \leftarrow \varphi(\vec{x},\vec{y})$ with the set of atoms occurring in $\varphi(\vec{x},\vec{y})$. We also denote $\vars(q)=\vec{x} \cup \vec{y}$ to be the set of all variables occuring in $q$.

Let $\Imc$ be an interpretation, $q(\vec{x})$ a CQ and $\vec{a}$ a tuple from
$\Delta^\Imc$ of length $|\vec{x}|$, we call $\vec{a}$ an \emph{answer to $q$
  in $\I$} and write $\Imc \models q(\vec{a})$ if
there is a map \[ \pi: \term(q) \mapsto \Delta^\Imc\] such that 
\begin{enumerate*}[label=(\it\roman*)]
	\item $\pi(\vec{x}) \,{=}\, \vec{a}$,
	\item $\pi(b) \,{=}\, b$ for each individual $b$,
	\item $\Imc \models P(\pi(\vec{z}))$ for each atom $P(\vec{z})$ in
          $q$, and
	\item  $\pi(t)=\pi(t')$ for each atom $t=t'$ in $q$.
\end{enumerate*}
The map $\pi$ is called a \emph{match} for $q$ in \I. We denote $ans(q,\I)$ to be the set of all answers
to $q$ in $\I$.

Let $(\T, \A)$ be a KB. 
A tuple of individuals $\vec{a}$ from $\ind{\A}$ with $|\vec{a}| = |\vec{x}|$
is a \emph{certain answer} of $q(\vec{x})$ over $\A$ wrt. \T if $\Imc \models
q(\vec{a})$ for all models $\I$ of $(\T,\A)$; 
$cert(q, \T , \A)$  denotes the set of certain answers of $q$ over  $\A$
wrt. \T. For queries  $q_1(\vec{x})$ and $q_2(\vec{x})$, we write $q_1(\vec{x})
\subseteq_{\T,\A} q_2(\vec{x})$ if $cert(q_1, \T , \A) \subseteq cert(q_2, \T
, \A)$.


\section{DL-Lite with Complex Role Inclusions}\label{sec:query-answering-dl}
In this section we study the restrictions required to add CRIs to \dllite in order to preserve its  nice
computational properties. In particular, we are interested in ensuring FO-rewritability, 
as well as a polynomial rewriting in the case of CQs. For this goal, a first restriction is to 
 assume a set $\nrsimple \subseteq \nr^\pm$ of \emph{simple roles}
closed w.r.t. inverses (i.e. $s \in \nrsimple$ implies $s^- \in
\nrsimple$); for each $r \in \nr \setminus \nrsimple$,  $r$ and $r^-$ are \emph{non-simple} roles. 
 We then define the extension of \dlLiteR with CRIs  as follows:

\begin{definition}[CRIs, \newDL]
	A \emph{complex role inclusion} (CRI) is an expression of the form $r \cdot s \ISA t $, with  $r,s,t \in \nr$. 
	
	A \newDL TBox $\T$ is a \dlliter TBox that may also contain CRIs such
	that: 
	\begin{itemize}
		\item For every CRI $r \cdot s \ISA t \in \T$, $s$ is simple and $t$ is non-simple. 
		\item If $s \ISA t \in \T$ and $t \in \nrsimple$, then $s \in \nrsimple$.
	\end{itemize}
	An interpretation $\I = (\Delta^\I, \cdot^\I)$ \emph{satisfies} a CRI 
	$r \cdot s \ISA t $ if for all $d_1,d_2,d_3 \in \Delta^\I$, $(d_1,d_2) \in
	r^\I$,  $(d_2,d_3) \in s^\I$ imply $(d_1,d_3) \in t^\I$. \
	\label{def:CRI}
\end{definition}
CRIs are a powerful extension of DLs, but unfortunately, their addition has a
major effect in the complexity of reasoning, and syntactic conditions such as
\emph{regularity} \cite{Kazakov2010AnEO} are often needed to preserve decidability. 
In the case of \dllite, even one single fixed CRI  $r \cdot e \ISA r$ destroys
first-order rewritability, since it can easily enforce $r$ to capture
reachability along  the edges $e$ of a given graph.  
%
\begin{lemma}\label{lemma:nloghard} \cite{Artale:2009}
	Instance checking in \newDL is 
	\textsc{NLogSpace}-hard in data complexity, already
	for 
	TBoxes consisting of the  CRI $r \cdot s \ISA r$ only.  
\end{lemma}

\subsection{Non recursive \newDL.} 
To identify FO-rewritable  fragments of \newDL 
it is natural to disallow $r \cdot s \ISA r$ by  restricting
cyclic dependencies between roles  occurring  in 
CRIs. 

\begin{definition}[\newDLnr TBoxes]	
	For a \newDL TBox $\T$, the \emph{recursion graph of \T} is
	the directed graph containing a node $v_A$ for each concept name $A$, and a
	node $v_r$ for each role name $r$ occurring in $\T$ and for each:
	\begin{compactitem}[-]
		\item $A_1 \ISA A_2 \in \T$, there exists an edge from $v_{A_2}$ to $v_{A_1}$;
		\item $t \ISA r \in \T$, there exists an edge from $v_{r}$ to $v_{t}$;
		\item $r \ISA p^- \in \T$, there exists an edge from $v_{p}$ to $v_{r}$;
		\item $A \ISA \exists r \in \T$, there exists an edge from $v_{r}$ to $v_{A}$;
		\item $\exists r \ISA A \in \T $, there exists an edge from $v_{A}$ to $v_{r}$;
		\item $r \cdot s \ISA t \in \T$, there exists an edge from $v_{t}$ to $v_{r}$ and to $v_{s}$.
	\end{compactitem}
	A role name $r$ is \emph{recursive in $\T$} if $v_r$ participates in a cycle 
	in the recursion graph of $\T$, and $t \cdot s \ISA r$ is
	\emph{recursive in $\T$} if $r$ is. 
	
	A \newDLnr TBox is a \newDL TBox \T  where no CRIs are recursive. 
\end{definition} 

Restricting CRIs to be non-recursive indeed guarantees
FO-rewritability. 

For a CQ $q$, we denote by 
$z^q$ an arbitrary but fixed variable not occurring in $q$; we will use such a
variable in the query rewriting rules through the rest of the paper. Additionally, 
we write $r^{(-)}(x,y) \in q$ if either $r(x,y) \in q$ or $r(y,x) \in q$;

\begin{definition}\label{def:spec}
	Let $\T$ be a \newDLnr TBox. 
	Given a pair $q,q'$ of CQs, we
	write $q \myspe q' $ whenever $q'$ is obtained by applying an
	\emph{atom substitution} $\theta$ or a \emph{variable substitution}
	$\sigma$ on $q$, where $\theta$ and $\sigma$ are as follows: 

		\begin{enumerate}[label=\bf S\arabic*., align=left, leftmargin=*]
			\item  \label{s:1} if $A_1\ISA A_2\in \T$ and $A_2(x)\in q$, then $\theta = [A_2(x) \mid A_1(x)]$;
			\item  if \label{s:2} $A\ISA \exists r\in \T$, $r(x, y)\in q$ and  $y$ is a non-answer variable occurring only once in $q$, then 
			$\theta = [r(x,y) \mid A(x)]$;
			\item  if \label{s:3} $\exists r \ISA A\in \T$ and $A(x)\in q$, then $\theta = [A(x) \mid r(x, z^q)]$;
			\item \label{s:4} if $r \ISA s \in \T$ and $s(x,y)\in q$, then $\theta = [s(x,y) \mid r(x,y)]$;
			\item \label{s:5} if $r \ISA s^{-}\in \T$ and $s(x,y)\in q$, then $\theta = [s(x,y) \mid r(y,x)]$;
			\item \label{s:6} if $t \cdot s \ISA r \in \T$ and $r(x,y)\in q$, then $\theta = [r(x,y) \mid \{t(x,z^q),s(z^q,y)\}]$;
			\item \label{s:7} if $x,y \in vars(q)$, then $\sigma= [x \mapsto y]$. 
		\end{enumerate}
	\smallskip \noindent
	We write $q \myspe^{*} q'$ if there is a finite sequence
	$q_0,\ldots,q_n$ of CQs such that $q=q_1,q'=q_n$, and $q_i {\myspe}
	q_{i+1}$ for all $0 \leq i < n$.
	\label{def:SpeTBox}
\end{definition}

By applying $\myspe$ 
exhaustively, we obtain a FO-rewriting of a given query $q$. 
\begin{definition} The \emph{rewriting of $q$ wrt. \T} is
	$\mathit{rew}(q,\T)= \{q\} \cup
	\{q'\mid q\myspe^* q'\}$.
\end{definition}
For any CQ $q$, $\mathit{rew}(q,\T)$ is a finite
query that can be effectively computed. 

\begin{lemma} 
	Let $\T$ be a \newDLnr TBox and let $q$ a CQ. 
	Each $q' \in \mathit{rew}(q, \T)$ is polynomially bounded in the size
	of $\T$ and $q$, and can be obtained in a polynomial number of steps. 
	\label{lemma:polybound}
\end{lemma}
\begin{proof}[Proof]
	Due to the non-recursiveness of the
	dependency graph and the restriction on simple roles, we show that we
	can assign to queries a (suitably bounded) degree that roughly corresponds to the
	number of rewriting steps that can be further applied. We 
	prove that for each $q'$ such that 
	$q \myspe^* q'$, the degree does not increase, and after
	polynomially many steps we will reach $q \myspe^* q''$ such that the
	degree strictly decreases. 
	
		We first define ${\cal G}_{acycl}$ as the acyclic version of the recursion graph of $\T$ in which nodes $n$ are labeled with a bag of predicates symbols, $bag(n)$, 
	and each maximal cycle in the recursion graph of $\T$ denotes a single node with a bag containing all predicates symbols participating in the cycle. 
	All other nodes are labeled with a bag consisting of single predicate symbol. The edges in ${\cal G}_{acycl}$ are obtained from the recursion graph, namely there is an edge between node $n$ and $n'$ if there exists an edge in the recursion graph between some $P \in bag(n)$ and $P' \in bag(n')$. 
	
	The function $mpath$ assigns a level to each node $n$ in ${\cal G}_{acycl}$ as follows:
	\begin{itemize}
		\item if $n$ has no outgoing edges, then $mpath(n)=0$;
		\item otherwise, $mpath(n)=max\{mpath(n')+1 \mid n \rightarrow n' \}$.
	\end{itemize}
	For a given query $q$, we define a function $\mathit{dgr}_{\T}(q)$ that,
	roughly,  bounds the number of rewriting steps that may be iteratively
	applied to it.  
	It is defined as follows:
	$$\mathit{dgr}_{\T}(q)= \underset{P(\vec{x}) \in q}{\sum} mpath(P)$$
	
	We will show that the application of the rules decreases the degree, except
	for some cases where the degree stays the same, but can only do so for 
	polynomially many rewriting steps (in the size of largest bag of ${\cal G}_{acycl}$).
	We show this bound before proving the main claim: 
	
	$(\ddagger)$ For each query of the form $q_1=q \cup \{P(\vec{x})\}$ such that
	$P$ participates in a cycle, then there are at most $k^2$ different queries of
	the form $q_2=q \cup \{P'(\vec{x'})\}$ that can be obtained by the rewriting
	rules and such that $\mathit{dgr}_{\T}(q_2)=\mathit{dgr}_{\T}(q_1)$,
	where $k$ is the size of the bag in ${\cal G}_{acycl}$ containing
	$P$ (unique bag containing $P$, since bags are triggered by largest cycles in $\T$).
	
	Let query $q_2$ be obtained by replacing $P(\vec{x})$ with $P'(\vec{x'})$ in
	$q_1$, where $\vec{x'}$ differes from $\vec{x}$ by at most one variable, $z^{q_1}$.
	Since $\mathit{dgr}_{\T}(q_2)=\mathit{dgr}_{\T}(q_1)$ it must be that $P,P'$
	belong to the same bag in ${\cal G}_{acycl}$.  
	If $P$ occurs in a cycle then, by the restriction of \newDL TBoxes, $P$ cannot be a non-simple role, 
	hence  $q_2$ is not obtained by applying {\bf S6}. Then, applying the axioms that
	trigger this cycle in $\T$, it must be that $q_2$ is obtained again after
	at most $k^2$ rewriting steps (number of distinct pairs of symbols in the bag). 
	
	Now that we have a bound on the number of times that the degree can stay the
	same for rewritings of a specific form, we can prove the lemma. 
	We will distinguish between the types of queries produced by the
	rules in Definition 3:
	\begin{enumerate}
		\item for rules {\bf S1-5}: $q \cup \{P(\vec{x}\}) \myspe q \cup \{P'(\vec{x'})\}$, and there is an arc between node labeled with $P$, and node labeled with $P'$, or they occur in the same bag in $\in{\cal G}_{acycl}$;
		\item for rule {\bf S6}: $q \cup \{r(x,y)\} \myspe q \cup \{t(x,z^q),s(z^q,y)\}$, where $z^q$ arbitrary fixed variable not occuring in $q$ and there exists arcs between the node labeled with $r$ and nodes labeled with $t$ and $s$;
		\item for rule {\bf S7:} $q(\vec{x}) \myspe \sigma(q(\vec{x}))$, where $\sigma$ is replaces one variable by another in $q$. 
	\end{enumerate}

	We now show that if $q_1 \myspe q_2$, then either \emph{(i)}
	$\mathit{dgr}_{\T}(q_2) < \mathit{dgr}_{\T}(q_1)$, or  \emph{(ii)} $\mathit{dgr}_{\T}(q_2) =\mathit{dgr}_{\T}(q_1)$ 
	if $q_1,q_2$ are as in $(\ddagger)$, and thus can only preserve the same degree for
	at most $k^2 \cdot |q_1|$ rewriting steps, or if $q_2$ is obtained by applying a substitution on $q_1$, which eventually leads to 
	a query with a unique variable.
	This will imply that, after at most $\mathit{dgr}_{\T}(q_1)\cdot (|q_1| \cdot k)^2$ steps,
	the degree will be zero and no more steps will be applicable.
	
	In what follows we show a proof by cases that matches cases 1--3 above.
	Firstly, for case $1$ above, if $P$ and $P'$ do not occur in same cycle, then $\mathit{dgr}_{\T}(q \cup \{P(\vec{x}\})) < \mathit{dgr}_{\T}(P'(\vec{x'}))$
	since there is an arc between $P$ and $P'$, hence $mpath(P) > mpath(P')$. The other subcase follows from $(\ddagger)$, therefore the degree of the queries obtained by  rules {\bf S1-5} decreases after at most $k^2 \cdot |q|$ rewriting steps.
	
	Next, we show for case $2$ above that: for each pair of queries  $q_1=q \cup \{r(x,y)\}$,  $q_2= q \cup \{t(x,z^q),s(z^q,y)\}$ such that $q_1 \myspe q_2$, we have that 
	$\mathit{dgr}_{\T}(q_2) < \mathit{dgr}_{\T}(q_1)$. 
	Since $q_2$ is obtaind by applying $t \cdot s \subsume r \in \T$, the claim
	follows immediately from the fact that $r$ cannot occur in a cycle and there
	must be an arc between node labeled $r$ and nodes labeled $t$ and
	$s$ in ${\cal G}_{acycl}$ ($s,t$ cannot belong to same bag due to the restriction of role inclusions between simple and non-simple roles).
	
	Lastly, for case $3$ above: for each pair of queries $q_1=q(\vec{x}) \myspe q_2=\sigma(q(\vec{x}))$, we have that $\mathit{dgr}_{\T}(q_2)=\mathit{dgr}_{\T}(q_1)$, however 
	in this case $vars(q_1) \subsetneq vars(q_2)$, hence such rule will eventually either reduce the size of $q_1$ and potentially make applicable 
	cases $1$ or $2$ hence  $dgr_{\T}(q_2) < dgr_{\T}(q_1)$ after at most $(k\cdot|q_1|)^2$ applications. 
	
	Therefore, we can conclude that each query $q' \in rew(\T,q)$ can be obtained after applying at most $\mathit{dgr}_{\T}(q)\cdot(|q|\cdot k)^2$ rewriting steps.
	
	We now argue the other part of the lemma, namely that each query in the rewriting has polynomial size. In case $1$ at most one new variable is introduced but the size of the query remains the same, and in case $2$ the size of the query increases by one, however only one of the newly introduced atoms (the non-simple role atom) may further trigger application of rule {\bf S6}, but only a polynomially bounded number of times, since the degree decreases. Therefore the size of each query in the rewriting is polynomially bounded. 
\end{proof}

The next result is shown analogously as for \dlLiteR \cite{CalvaneseGLLR07},
extended for new rule ${\bf S6}$. The full proof can be found in the appendix.

\begin{lemma}\label{lemma:PerfectRef}
	Let $\T$ be a \newDLnr TBox, $q$ a CQ. For every ABox $\A$ consistent with
	$\T$:
	\[cert(q,\T,\A) = \underset{q' \in \mathit{rew}(q,\T)}{\bigcup}
	cert(
	q',\emptyset,\A).\]
\end{lemma}

Non-recursive CRIs preserve FO-rewritability of \dllite, but their addition 
is far from harmless. Indeed, unlike the extension with transitive roles, 
even non-recursive CRIs increase the complexity of testing KB consistency.

\begin{theorem}
	Consistency checking in \newDLnr is coNP-complete. 
\end{theorem}

\begin{proof}
	\emph{Upper-bound:} Similarly as for standard \dllite, 
	inconsistency checking can be reduced to UCQ
	answering, using a CQ $q_\alpha$ for testing whether each
	disjointness axiom $\alpha$ is violated. By  Lemmas \ref{lemma:polybound} and
	\ref{lemma:PerfectRef}, an NP procedure can guess one such
	$q_\alpha$, guess a $q_\alpha'$ in its rewriting, and evaluate $q_\alpha'$ over
	$\A$.
	
	\emph{Lower-bound:} We reduce the complement of 3SAT  to KB satisfiability. 
	Suppose we are given a conjunction $\varphi= c_1 \land \dots \land c_n$ of clauses of the form 
	$\ell_{i_1} \lor \ell_{i_2} \lor \ell_{i_3}$, where the $\ell_k$ are literals, i.e., propositional variables or their negation. 
	Let $x_0, \dots, x_m$ be all the propositional variables occurring in $\varphi$.

	In order to encode the possible truth assignments of each variable  $x_i$, we take two fresh roles $r_{x_i}$ and ${\bar{r}}_{x_i}$, intended to be disjoint. 
	We construct a \newDLnr TBox $\T_{\varphi}$ containing, for every  $0 \leq i \leq m$, the following axioms:
	\begin{align*}
		\mathbf{disj}(r_i, &{\bar{r}}_{x_i}), &A_i &\subsume \exists r_{x_i} \sqcap \exists {\bar{r}}_{x_i}, &\exists r_{x_i}^- &\subsume A_{i+1},   \\
		\exists ({\bar{r}}_{x_i})^-  &\ISA A_{i+1}, & r_{x_i} &\subsume t,  & {\bar{r}}_{x_i} &\subsume t
	\end{align*}
	These axioms have a model that is a full binary tree, rooted at $A_0$ and
	whose  edges are labeled with the role $t$, and with different combinations
	of the roles  $r_i$ and ${\bar{r}}_{x_i}$. Intuitively, each path represents
	a  possible variable truth assignment. Further, $\T_\varphi$ contains axioms relating each variable assignment with the clauses it satisfies, using roles $s_{c_1}, \dots, s_{c_n}$. More precisely, we have the following role inclusions for $0\leq i \leq m$, and $1\leq j\leq n$:
	\begin{align}
		r_{x_i} &\subsume s_{c_j}, \quad  \text{ if } x_i \in c_j  & 
		\bar{r}_{x_i} &\subsume s_{c_j},  \quad \text{ if } \neg x_i \in c_j \label{eq:4}
	\end{align}
	To encode the evaluation of all clauses, we have axioms propagating down the tree all clauses  satisfied by some assignment. Note that we could do this easily using a CRI such as $s_{c_j} \cdot t 
	\ISA s_{c_j}$. However, this would need a recursive role $s_{c_j}$. Since the depth of the assignment tree is bounded by $m$, we can encode this (bounded) propagation using at most $m$ roles $s^i_{c_j}$ ($1\leq i \leq n$) for each clause $c_j$, which will be declared as subroles of another role $s^*_{c_j}$. 
	For  $1 \leq j \leq n$ and $1 \leq i < m$,  we have the following CRIs: 
	\begin{align*}
		s_{c_j} \cdot t &\subsume s^1_{c_j}  &s^{i}_{c_j} \cdot t &\subsume s^{i+1}_{c_j} &	s^i_{c_j} &\subsume s^*_{c_j}
	\end{align*}
	Thus, if  $c_j$ is satisfied in a $t$-branch of the assignment tree,  its leaf will have an incoming  $s^*_{c_j}$ edge. Now, in order to encode that there is at least one clause that is not satisfied, we need to forbid the existence of a leaf satisfying the concept $\exists (s^*_{c_1})^- \sqcap \dots \sqcap \exists (s^*_{c_n})^-$.  This  cannot be straightforwardly written in \newDLnr, 
	but we resort again to CRIs to propagate information: 
	\begin{align}
		\exists (s^*_{c_1})^- & \subsume \exists t_1 &
		s^*_{c_k} \cdot t_1 &\subsume p_k^1, \, 2 \leq k \leq n 
	\end{align}
	Next, for $2 \leq i \leq n, \, i < k \leq n$ we have the following:
	\begin{align}
			\exists (p_i^{i-1})^- &\subsume \exists t_i &
		p_k^{i-1} \cdot t_i &\subsume p_k^i
	\end{align}
	By adding the axiom
	$\exists t_n \subsume \bot$, we obtain the required restriction. 
	In the appendix we prove that $\varphi$ is
	unsatisfiable iff 
	$(\T_\varphi,\{ A_0(a)\})$ is satisfiable.
\end{proof}

\begin{theorem} 
	CQs over \newDLnr KBs are FO-rewritable. The complexity of answering CQs over consistent \newDLnr KBs is in $AC_0$ in data, and \textit{NP-complete} in combined complexity.
\end{theorem}

The FO-rewritability and data complexity follow from Lemma \ref{lemma:PerfectRef}, while the  NP-hardness in combined complexity is inherited from  CQs over plain relational databases. The NP membership follows from Lemma \ref{lemma:polybound} and the fact that 
guessing a rewriting, it is possible to verify in polynomial time if it has a match over the ABox.

\subsection{Recursion-safe \newDL} 

Additionally to the increased complexity, 
\newDLnr has another relevant limitation:
it cannot express  CRIs like $\role{occursIn} \cdot \role{locatedIn} \subsume
\role{occursIn} $ as we need in our motivating example. 
We introduce another extension of \newDL that allows for 
CRIs with some form of controlled recursion. 

\begin{definition}[Recursion safe \newDL]
In a \emph{recursion safe} \newDL TBox all CRIs $r_1\cdot s \ISA r_2 \in \T$ satisfy: 
\begin{itemize}
	\item If $r_2$ participates in some cycle  in the recursion graph of $\T$, then the cycle has length at most one, and $r_1 = r_2$.   
	\item There is no axiom of the form $A \ISA \exists t \in \T$ with $t
          \ISA_\T^{\sf s} s $ or   $t \ISA_\T^{\sf s} s^- $, 
where $\ISA_\T^{\sf s}$ denotes the reflexive and transitive
closure of the simple inclusions in a  \newDL TBox $\T$, that is, of the relation $s_1 \ISA s_2  \in \T$
with  $s_2 \in \nrsimple$.
	\end{itemize}

\end{definition}

The key idea behind  recursion safety is that every 
recursive CRI is `guarded' by a simple role that is not existentially 
implied. For query answering, we can assume that only ABox individuals are
connected by these guarding roles, and thus CRIs only `fire' close to the
ABox (that is, each
pair in the extension of a recursive roles has at least one individual).  In fact, we show below that every consistent recursion-safe KB has a model
where both conditions hold.

\begin{example}
$\K_e$ is recursion safe, since
$ \role{occursIn} \cdot \role{locatedIn} \ISA \role{occursIn}$ is the only CRI,
and $\role{locatedIn}$ 
is not implied by any existential axiom in 
  $\T_e$.
\end{example}


\subsubsection{Reasoning in recursion safe \newDL.}

Standard reasoning problems like consistency checking and answering instance
queries are tractable for recursion safe \newDL KBs. In fact, for a given KB,  
we can build a polynomial-sized interpretation that is a model whenever the KB
is consistent, and that can be used for testing entailment of assertions and
of disjointness axioms.


\begin{definition}
	Let $(\T,\A)$ be a recursion safe \newDL KB. We define an
        interpretation $\Smod{\T}{\A}$ 
as follows. 
As domain we use the individuals in $\A$, 
fresh individuals $c_{ar}$ that serve as $r$-fillers for individual $a$, 
and fresh individuals $c_r$ that serve as shared $r$-fillers for the objects
that are not individuals in $\A$. That is, $\Delta^{\Smod{\T}{\A}} = D_0 \cup
D_1 \cup D_2$, where 
		\begin{align*}
	&D_0 = ind({\A}), \\
	&D_1 = \{ c_{ar} \mid a \in D_0, r \text{ occurs on the rhs of a CI in }\T \}, \\
	&D_2 = \{ c_r \mid  r \text{ occurs on the rhs
          of a CI in }\T  \}.
\end{align*}
   The interpretation function has 
   $a^{\Smod{\T}{\A}} = a$ for each $a \in \Delta^{\Smod{\T}{\A}}$, and
   assigns to each concept name $A$ and each role name $r$ in $\Sigma_{\T}$
   the minimal set of the form $A^{\Smod{\T}{\A}} \subseteq
   \Delta^{\Smod{\T}{\A}}$, $r^{\Smod{\T}{\A}} \subseteq
   \Delta^{\Smod{\T}{\A}} \times \Delta^{\Smod{\T}{\A}}$ such that the
   following conditions hold, for all $A \in \nc$, $B$ a basic concept, and $r,r_1,r_2,s,t \in \nr$:
	\begin{enumerate}
		\item  if $A(a) \in \A$ then $a \in A^{\Smod{\T}{\A}}$, 
 and if $r(a,b) \in \A$ then $(a,b) \in r^{\Smod{\T}{\A}}$.
		\item  If $B \ISA \exists r \in \T$, $a \in B^{\Smod{\T}{\A}} \cap D_0 
                  $ then $(a, c_{ar}) \in
                  r^{\Smod{\T}{\A}}$.
		\item  if $B \ISA \exists r \in \T$, $d \in B^{\Smod{\T}{\A}} \cap (D_1 \cup D_2)$ then  $(d, c_r) \in r^{\Smod{\T}{\A}}$.
		\item  if $B \ISA A \in \T$, $d \in  B^{\Smod{\T}{\A}}$
                  then $d \in A^{\Smod{\T}{\A}}$. 
		\item if  $r_1 \ISA r_2 \in \T$,  $(a,b) \in r_1^{{\Smod{\T}{\A}}}$ then $(a,b) \in r_2^{{\Smod{\T}{\A}}}$.
		\item if $r_1 \ISA r_2^- \in \T$, $(a,b) \in r_1^{{\Smod{\T}{\A}}}$ then $(b,a) \in r_2^{{\Smod{\T}{\A}}}$.
		\item  if $r \cdot s \ISA t \in \T$, $(a,b) \in r^{{\Smod{\T}{\A}}}$ and $(b,c) \in s^{{\Smod{\T}{\A}}}$ then
		$(a,c) \in t^{{\Smod{\T}{\A}}}$.
	\end{enumerate}
\label{def:ABoxExt}
\end{definition}

For $\Smod{\T}{\A}$, we can show the following useful
properties:

\begin{proposition}
	Let $\T = \T_p \cup \T_n$ be a recursion safe \newDL TBox, where $\T_p$ contains only positive inclusions, and $\T_n$ contains only disjointness axioms. Then, for every ABox $\A$:
	\begin{enumerate}
		\item [$\bf P1$] \label{C1} If $(\T,\A)$ is consistent, then
                  $\Smod{\T}{\A} \ent (\T, \A)$. 
	\item [$\bf P2$] \label{C2}  $(\T,\A)$ is inconsistent iff  $\Smod{\T}{\A}
                  \not\models \alpha$ for some $\alpha \in \T_n$.
		\item [$\bf P3$] \label{C3} If $(\T,\A)$ is consistent and $q$ is an
                  instance query, then $cert(q,\T,\A) = ans(q,\Smod{\T}{\A})$. 
	\end{enumerate}
\label{prop:recsafeABoxExt}
\end{proposition}
\newcommand{\tp}{\mathsf{tp}}
\begin{proof}[Proof (sketch)] 
  To prove  $\bf P1$, we assume that $(\T,\A)$ is consistent. 
Verifying that  $\Smod{\T}{\A}$ satisfies all but the disjointness axioms is
easy from the definition of $\Smod{\T}{\A}$. 
Let $\I$ be an arbitrary model of $(\T,\A)$. 
For  $d,d' \in \Delta^{\I}$, let $\tp_{\I} (d) = \{B \mid d \in B^{\I}\}$ the set of basic concepts satisfied at $d$ in $\I$, and
$\tp_{\I}(d,d') = \{r \mid (d,d') \in r^{\I} \}$, the set of roles connecting $d$ and $d'$ in $\I$. 
The following claim shows a key property of $\Smod{\T}{\A}$. The proof of the claim can be found in the appendix.
	\begin{claim}
	 For any given $d \in \Delta^{\Smod{\T}{\A}}$ \emph{(i)} there exists $e \in \Delta^{\I}$ such that $\tp_{\Smod{\T}{\A}}(d) \subseteq \tp_{\I}(e)$ and  \emph{(ii)} for each $d' \in \Delta^{\Smod{\T}{\A}}$ such that $\tp_{\Smod{\T}{\A}} (d,d') \neq \emptyset$ we have that there exists $e' \in \Delta^{\I}$ such that $\tp_{\Smod{\T}{\A}} (d,d') \subseteq \tp_{\I} (e,e')$.
	 \label{pr:1}
	\end{claim}
Towards a contradiction, assume there is $\alpha = {\bf disj}(B_1,B_2) \in \T$
such that $\Smod{\T}{\A} \not \models \alpha$; the case of role disjointness
axioms is analogous.  Then there is $d \in \Delta^{\Smod{\T}{\A}}$ with $B_1,
B_2 \in \tp_{\Smod{\T}{\A}}$, and by the claim above, $B_1,B_2 \in \tp_{\I}(d)$
for each model $\I$.  Hence $\Smod{\T}{\A} \ent \alpha$, and this concludes
proof of $\bf P1$. 
Properties
$\bf P2$ and $\bf P3$ can also be shown using the above Claim \ref{pr:1}  and
the fact that $\Smod{\T}{\A}$ is a model of the KB. 
\end{proof}

This proposition allows us to establish the following results: 

\begin{theorem} 
For recursion safe \newDL KBs, consistency checking and instance query
answering are feasible in polynomial time in combined complexity. 
\end{theorem}

The recursion safe fragment of \newDL is not FO-rewritable: indeed,
the TBox in the proof of Lemma \ref{lemma:nloghard} is recursion safe.  
However, we can get rid of recursive CRIs and regain rewritability if we 
have guarantees that they will only be relevant on paths of bounded length. 
We formalize this rough intuition next. 
\begin{definition}[k-bounded ABox]
	Let $\T$ be a 
    \newDL TBox and $\A$ an ABox.
	Let $S$ be a set of simple roles. 
    Given $a,b \in ind(\A)$, we say that there exists an
        \emph{$S$-path of length $n$ between $a$ and $b$ (in $\A$ w.r.t.\,\T)} 
 if there exist pairwise distinct 
	$d_1, \dots,d_{n-1} \in ind(\A)$ with $d_i \not\in \{a,b\}$, and 
     $s_1 (a,d_1),$ $\dots,$ $s_i (d_{i-1},d_i),$ $\dots,$ $s_n (d_{n-1},b) \in \A$
     such that $s_i \ISA_{\T}^{\sf s} s$ and $s \in S$, $1 \leq i <n$. 
	Let $S_r = \{s \mid r \cdot s \ISA r \in \T \}$. We say that $\A$ is
        \emph{$k$-bounded for $\T$} if  for each recursive $r \in \T$
        there is no $S_r$-path 
   of size larger than $k$. 
\end{definition}
We simulate recursive CRIs by unfolding them into
$k$ non-recursive ones. 
\begin{definition}[$k$-unfolding, $k$-rewriting]
Let $\T$ be an arbitrary recursion-safe \newDL TBox. For any $k \geq 0$, a \emph{$k$-unfolding of $\T$}
is a \newDLnr TBox $\T_k$ obtained by replacing each $r \cdot s \ISA r \in \T$
with the 
axioms 	\[ r_{j-1} \cdot s \subsume r_{j}  \quad 	 r \subsume r_0  \quad  r_j
\subsume \hat{r} \] 
for $1 \leq j \leq k$, where $\hat{r}$ and each $r_j$ are
fresh role names. 
 For any CQ $q$ over $\Sigma_{\T}$, let $\hat{q}$ be the
query obtained from $q$ by replacing with $\hat{r}(x,y)$ each 
 occurrence of $r(x,y)$, where $r$ is a recursive role in $\T$.
We call $rew(\hat{q}, \T_k)$ the $k$-rewriting of $q$ w.r.t.\,
$\T$. 
\end{definition}

For $k$-bounded ABoxes, the $k$-rewriting is a complete FO-rewriting:

\begin{lemma}
	Let $\T$ be a recursion safe \newDL TBox, $\T_k$ a $k$-unfolding of $\T$, for some $k \geq 0$, and $q$ a CQ over $\Sigma_{\T}$.
	Then, for every $k$-bounded ABox $\A$: 
	$$cert(q,\T,\A) = \underset{q' \in rew(\hat{q}, \T_k) }{\bigcup}cert(q',\emptyset, \A) $$
	\label{claim:kboundFO}
\end{lemma}

\begin{proof}[Proof (sketch)] 
In a nutshell, recursion-safety ensures that recursive CRIs in $\T$ can only
`fire' in the chase along $S_r$-paths in the ABox. If 
 $\A$ is $k$-bounded for $\T$, then such paths  have length $\leq k$, so we
 get  that every pair $(d,d')$ that should be added to a
recursive role $r$ is added to some $r_j$, and hence to $\hat{r}$ (see appendix for full proof).
\end{proof}




\section{Ontology-driven Query Reformulations }\label{sec:queryreform}


In this section we focus on leveraging \newDL ontologies for modifying queries to either decrease or increase the set of answers. 
Such query operations semantically denote query containment, over any data set.

\subsection{Query Restraining}
A closer look to the query rewriting rules of \newDLnr, allows us to obtain TBox-driven 
query modifications, and it follows from Lemma \ref{lemma:PerfectRef} that each such query 
is more restrictive, meaning that it potentially has less answers.

\begin{definition}\label{def:spec}
  Let $\T$ be a \newDL TBox. 
  Given a pair of CQs $q,q'$ , we
  write $q \myspeS q' $ if $q \myspe q'$.
\smallskip \noindent
 We write $q \myspeS^{*} q'$ if there is a finite sequence
  $q_0,\ldots,q_n$ of CQs such that $q=q_1,q'=q_n$, and $q_i\myspeS
  q_{i+1}$ for all $0 \leq i < n$, and call $q'$ 
a \emph{restraining of } $q$ w.r.t. $\T$ if $q \myspe^{*} q'$.
\label{def:SpeTBox}
\end{definition}

We say that the application of any of the rules \ref{s:1}-\ref{s:7} \emph{restrains} a query, because
the answers of the resulting query are necessarily contained in the answers of
the original one. 

\begin{example}
	 Suppose we want to query all events occurring in some city, over $\T_e' = \T_e \cup \{\role{occursIn} \cdot \role{locatedIn} \ISA \role{occursIn}\}$:
	 $$q(x)\leftarrow\conc{Event}(x), \role{occursIn}(x,y), \conc{City}(y)$$
	 Then by rule {\bf S1} applied for axiom $\conc{CulturEvent} \subsume \conc{Event}$ in previous definition we can obtain
	 $$q_1(x)\leftarrow\conc{CulturEvent}(x), \role{occursIn}(x,y),\conc{City}(y)$$
	 We next apply {\bf S6} for $\role{occursIn} \cdot \role{locatedIn} \ISA \role{occursIn}$  to obtain
	 \begin{align*}
	 	q_2(x)\leftarrow \, &\conc{CulturEvent}(x), \role{occursIn}(x,z^{q_1}), \role{locatedIn}(z^{q_1},y), \conc{City}(y).
     \end{align*}\demo
   \label{ex:CRITBox}
\end{example}

Such reformulations can be evaluated efficiently when the TBox is recursion-safe and the ABox is $k$-bounded. The following Proposition follows from Lemma \ref{claim:kboundFO}.

\begin{proposition}
		Let $\T$ be a recursion safe \newDL TBox. For any two CQs, such that $q_1 \myspeS q_2$ we have that 
	$cert(q_2, \T, \A) \subseteq cert(q_1, \T, \A)$, for any $k$-bounded ABox $\A$.
	\label{prop:speContTBox}
\end{proposition}



\subsection{Query Relaxation}\label{sec:query-reformulation}

We have seen that the query reformulation rules that `apply' the axioms in a
right-to-left fashion, provide natural means to restrain queries. 
The natural next step is to define analogous rules that use the axioms in a
left-to-right fashion, to relax queries. Note that in the next definition,
rules \ref{g:1}--\ref{g:6} are, essentially, the dual of rules 
 \ref{s:1}--\ref{s:6}; 
Rule \ref{g:7} is a bit different, since it simply allows us to relax a query
by dropping atoms. 



%

 %
 

\begin{definition}
	Let $\T$ be a \newDL TBox. Given a pair $q,q'$ of CQs, we
	write $q \mygen q' $ whenever $q'$ is obtained from $q$ by applying atom substitution $\theta$ as follows:
	\begin{enumerate} [label=\bf G\arabic*., align=left, leftmargin=*]
		\item\label{g:1} if $A_1\ISA A_2\in \T$, and $A_1(x)\in q$, then $\theta = [A_1(x) \mid A_2(x)]$;
		\item\label{g:2} if $A\ISA \exists r\in \T$, and $A(x)\in q$,  then $\theta = [A(x)\mid r(x,z^q)]$;
		\item\label{g:3} if $\exists r \ISA A\in \T$,  $r(x,y)\in q$ and $y$ is a non-answer variable occurring only once in $q$, then $\theta = [r(x,y) \mid A(x) ]$;
		  \item\label{g:4} if $r \ISA s \in \T$ and  $r(x,y)\in q$, then $\theta=[ r(x,y) \mid s(x,y) ]$;
		\item \label{g:5} if $r \ISA s^{-}\in \T$, and $r(x,y)\in q$, then $\theta=[r(x,y) \mid s(y,x)]$;
		\item\label{g:6} if $r \cdot s \ISA r \in \T$, $r(x,y), s(y,z) \in q$ and $y$ is a non-answer variable that does not occur elsewhere in $q$, 
	then $\theta = [\{r(x,y), s(y,z)\} \mid r(x,z) ]$;
		\item\label{g:7} if $A(x) \in q$ and $x$ is a non-answer variable, then $\theta = [A(x) \mid \emptyset]$.
	\end{enumerate}
\smallskip \noindent
	We write $q \mygen^{*} q'$ if there is a finite sequence
	$q_0,\ldots,q_n$ of CQs such that $q=q_1,q'=q_n$, and $q_i \mygen
	q_{i+1}$ for all $0 \leq i < n$. 
	We call $q'$ a \emph{query relaxation of} $q$ w.r.t $\T$ whenever $q \mygen^{*} q'$. 
\end{definition}

\begin{example}
	Considering the following query over $\T'_e$:
	\begin{align*}
	q(x)\leftarrow \,&\conc{Concert}(x), \role{occursIn}(x,y), \role{locatedIn}(y,z), 
	z={\sf Vienna}
	\end{align*}
	We can apply rule {\bf G1} with axiom $\conc{Concert} \subsume \conc{CulturEvent}$ and obtain 
	query:
		\begin{align*}
	q_1(x)\leftarrow&\conc{CulturEvent}(x), \role{occursIn}(x,y), \role{locatedIn}(y,z), 
	z={\sf Vienna}
	\end{align*}
	Further, by applying {\bf G6} with axiom $\role{occursIn \cdot locatedIn} \subsume \role{occursIn}$ 
	we obtain 
		\begin{align*}
	q_2(x)\leftarrow&\conc{CulturEvent}(x), \role{occursIn}(x,z), z={\sf Vienna}
	\end{align*}\demo
\end{example}

The following result is the analogous of Proposition~\ref{prop:speContTBox}.
 
\begin{proposition}
	Let $\T$ be a recursion safe \newDL TBox. For any two CQs, such that $q_1 \mygen q_2$ we have that 
	$cert(q_1,\T,\A) \subseteq cert(q_2, \T,\A)$, for any $k$-bounded ABox $\A$. 
\end{proposition}

\begin{proof}
	The claim clearly holds whenever $q_2$ is obtained from $q_1$ using rules {\bf G1-2, G4-5, G7}. 
	In case of {\bf G3} and {\bf G6}, the replacement of query atoms results in dropping a variable in $q_1$.
	Since such variable is not an answer-variable and does not occur elsewhere in $q_1$, the replacement which is justified by an axiom in $\T$,
	does not disconnect terms of $q_1$.
\end{proof}

\section{Data-driven Query Reformulations} \label{sect:datadriven}
In this section we are interested in characterizing query modifications which are data dependent, therefore the containement 
relation holds only for the current dataset. 
The ontology-driven query reformulation rules in the
previous section may not capture all  the query variations that the user is
interested in. 
For instance, they do not allow to generalize a query about concerts
in Vienna, to a query  asking for all concerts in Austria, or
to specialize to one about concerts at the State Opera. 
Clearly, such reformulations cannot be done on the basis of the TBox alone,
since they consider the specific dataset, more specifically, they are based
on the assertions  $\role{locatedIn}(\sf Vienna, \sf Austria)$ and 
$\role{locatedIn}(\sf StateOper, \sf Vienna)$. 

We may also be interested in some reformulations that do not consider 
specific instances, but they are based on \emph{dependencies} that hold in our
dataset rather than on TBox axioms.   
In our running example, a quick inspection at the data in Figure \ref{fig:eventData}
tells us that \textit{every existing venue is located in a city}. 
We could use this information to relax the query 
\begin{align*}
q(x)\leftarrow & \conc{Event}(x), \role{occursIn}(x,y), \conc{Venue}(y) \mbox{\qquad
  into}\\
q'(x) \leftarrow &\conc{Event}(x), \role{occursIn}(x,y), \role{locatedIn}(y,z), \conc{City}(z)
\end{align*}
We note that such a reformulation could be done with using rules  \ref{s:1}-\ref{s:7}, if 
we had an inclusion $\conc{Venue} \ISA \exists\role{locatedIn}.\conc{City}$  
in the TBox\footnote{This is not in the syntax of \dlliter, but easily
  expressible.}. However, we may not have such an axiom, and it may not
be possible or desirable to add it. 
For that reason, we allow to also reformulate the query using some
containments that are not guaranteed by the TBox, but we can test that they 
hold for the dataset being considered. For that, we write $q_1 \subseteq_\K q_2$ to denote that $cert(q_1, \K) \subseteq cert(q_2, \K)$.


\newcommand{\myspeP}{{\leadsto_{\K}^s}}
\newcommand{\mygenP}{{\leadsto_{\K}^g}}

\begin{definition} 
	Let $\K = (\T, \A)$ be a \newDL KB.
	\begin{itemize}
		\item Given a pair $q, q'$ of CQs, we
	write $q \myspeP q' $ if $q \myspeS q'$ or $q'$ is obtained from $q$ using atom substitution $\theta$ as follows:
	\begin{enumerate}[label=\bf SD\arabic*, align=left, leftmargin=*]
		\item\label{sd:1} if $A(x) \in q$, $\K\models A(a) $, then $\theta = [ \emptyset \mid x=a ]$;
		\item\label{sd:2} if $r(x,y) \in q$ and $\K \models r(a,b) $, then  either (i) $\theta= [\emptyset \mid x=a ]$ or (ii) $\theta= [\emptyset \mid x=b]$; 
\item \label{sd:3} if $A_1(x)\subseteq_\K A_2(x)$ and $A_2(x) \in q$ then $\theta= [A2(x) \mid A_1(x)]$;
\item\label{sd:4} for $q^*(x) \leftarrow r(x,y),A'(y)$ if $q^*(x)\subseteq_\K A(x)$ and $A(x) \in q$, then $\theta = [A(x) \mid \{r(x,z^q),A'(z^q)\}]$;
\item\label{sd:5}  for $q^*(x) \leftarrow r(x,y),A'(y)$ if  $A(x)\subseteq_\K q^*(x)$, and $r(x,y),A'(y) \in q$ such that $y$ does not occur elsewhere in $q$, then $\theta = [\{r(x,y),A'(y)\} \mid A(x)]$; 
\item \label{sd:6} for $q^*(x )\gets r(x,y),A(y)$ and $\hat{q}(x)\gets
                p(x,y),A'(y)\quad$, if $\hat{q}(x)\subseteq_\K
                q^*(x)$ and $r(x,y),A(y) \in q$ such that $y$ does not occur elsewhere in $q$, then $\theta = [\{r(x,y),A(y)\} \mid \{p(x,z^q),A'(z^q)\} ]$.
	\end{enumerate}
	\item  Given a pair $q, q'$ of CQs, we
	write $q \mygenP q' $ if $q \mygen q'$ or $q'$ is obtained from $q$ using atom substitution $\theta$ as follows:
	\begin{enumerate}[label=\bf GD\arabic*, align=left, leftmargin=*]

        \item\label{gd:1} if $x=a \in q$ and $\K\models A(a)$, then $\theta = [ x=a \mid A(x)]$;

        \item\label{gd:2} if $x=a \in q$ and $\K\models r(a,b) $, then $\theta = [x=a \mid \{r(x,y),y=b\}]$;
 
        \item\label{gd:3} if $A_1(x)\subseteq_\K A_2(x)$ and $A_1(x) \in q$, then $\theta=[A_1(x) \mid A_2(x)]$;

        \item\label{gd:4} for $q^*(x) \gets r(x,y),A'(y)$, if $A(x)
          \subseteq_\K q^*(x)$ and $A(x) \in q$, then $\theta = [A(x) \mid \{r(x,z^q),A'(z^q)\}]$; 

        \item\label{gd:5} for $q^*(x) \gets r(x,y),A'(y)$, if $ q^*(x)
          \subseteq_\K A(x)$ and $r(x,y),A'(y) \in q$ such that $y$ does not occur elsewhere in $q$, then $\theta = [\{r(x,y),A'(y)\} \mid A(x)]$

        \item\label{gd:6} for $q^*(x)= r(x,y),A(y)$ and
          $\hat{q}(x)\gets P(x,y),A'(y)\quad$, if $ q^*(x)
          \subseteq_\K \hat{q}(x)$ and $r(x,y),A(y) \in q$ such that $y$ does not occur elsewhere in $q$, then $\theta = [\{r(x,y),A(y)\}\mid \{p(x,z^q),A'(z^q)\}]$.
	\end{enumerate}
\end{itemize}
For $\delta \in \{g,s\}$, we write $q \gsArrowT{\K}{\delta} q' $ if there is a finite sequence
  $q_0,\ldots,q_n$ of CQs such that $q=q_1,q'=q_n$, and $q_i \gsArrow{\K}{\delta}
  q_{i+1}$ for all $0 \leq i < n$.

  We call $q'$ a \emph{(data-driven) restriction} of $q$ w.r.t. $\K$ if $q
  \gsArrowT{\K}{s} q'$, 
   and we call $q'$ a \emph{(data-driven) relaxation} of $q$ w.r.t. $\K$ if
  $q \gsArrowT{\K}{g} q'$. 
\end{definition}

In a nutshell, we have two kinds of rules: those that use assertions, and
those that use inclusions $q_A(x) \subseteq_\K q_B(x)$. 

In the first group, we have specialization rules \ref{sd:1} and \ref{sd:2}, and generalization
rules \ref{gd:1} and \ref{gd:2}.
If our query contains an atom $A(x)$ and we know that $a$ is an instance of
$A$ (that is, $\K \models A(a)$), then we can
specialize the query by making $x$ equal to $a$ (\ref{sd:1}); similarly, 
if $r(x,y)$ is in $q$ and we have $\K \models r(a,b)$, then we can
add either $x = a$ or $y = b$ (\ref{sd:2}). 
In the converse direction, if the query is equating some variable $x$ to a
constant $a$  that is an instance of $A$, then we can replace $x=a$ with
$A(x)$, and generalize the query by allowing $x$ to be any instance of $A$,
rather than just $a$ (\ref{gd:1}). Similarly we can use role assertions, and 
if $x=a$ is in $q$ and $r(a,b) \in \A_c$, then we can replace $x=a$ with the
pair $r(x,y),y=b$ (\ref{gd:2}). 

\begin{example} 
Using (\ref{gd:2}) and the assertion \[ \role{locatedIn}({\sf
  Vienna},{\sf Austria})\]
 we can relax the query 
$$q(x)\leftarrow\conc{Concert}(x), \role{occursIn}(x,y),y={\sf Vienna}$$ 
obtaining the following query:
\begin{align*}q'(x)\leftarrow &\conc{Concert}(x), \role{occursIn}(x,y), \role{locatedIn}(y,z^{q}), z^q = {\sf Austria}\end{align*}
\noindent
That is, we relax the query from the concerts in Vienna, to those that occur in
an Austrian city. \demo
\end{example}

The second group of rules,  (\ref{sd:3}--\ref{sd:6}) and
(\ref{gd:3}--\ref{gd:6}) are very similar to the ones in the previous section, 
however, now they allow to replace 
$B(x)$ by $A(x)$ in a specialization, 
not only when $A \ISA B$ is in $\T$, but also when the weaker 
condition $A(x) \subseteq_\K B(x)$ holds.  
Such replacements are also allowed for some more complex pairs of atoms. For
example, if  $r(x,y),{B}(y)  \subseteq_\K  A(x)$, $A(x)$ can be replaced  by
$r(x,y),B(y)$ to specialize the query. This would be similar to a rule for (non-\dllite) 
  axioms of the form  $\exists{r}.{B} \ISA A$ in
  Definition~\ref{def:spec}.

We remark that in this second group of rules,  defined in terms of inclusions $q_A(x) \subseteq_\K q_B(x)$, 
  the queries $q_A$ and $q_B$ have a restricted shape, with at most two atoms and two variables. 
Testing for  their containment is not 
an expensive task for recursion safe \newDL TBox and $k$-bounded ABox, since query answering in this case is FO-rewritable. 
Moreover, the search space of all queries that could result in applicable rules is polynomially bounded in the input. 

\begin{example}
Consider 
  $$q(x)\leftarrow\conc{Event}(x), \role{occursIn}(x,y), \conc{City}(y)$$ 
If $\conc{City}(x) \subseteq_\K \exists \role{locatedIn}.\conc{Country}(x)$,
then using rule \ref{gd:4}  we obtain:
  \begin{align*}
  q'(x)\leftarrow&\conc{Event}(x), \role{occursIn}(x,y), \role{locatedIn}(y,z^q), \conc{Country}(z^q)
  \end{align*}
Note that, after this rule application, we can actually apply the rule
\ref{g:6}  and obtain  
  \begin{align*}
  q'(x)\leftarrow&\conc{Event}(x), \role{occursIn}(x,z^q), \conc{Country}(z^q)
  \end{align*}
This illustrates that our data-driven rules are useful for query reformulation
not only on their own, but also because they may allow  other relevant
reformulations that were not applicable otherwise. 
\demo
\end{example} 
We show that our rules  indeed relax and restrain queries. Note that in
this case,  the answers containment only holds when evaluated over
$(\A,\T)$, but not for an arbitrary~$\A$. 
\begin{proposition}
Let $\K= (\T,\A)$ be a KB where $\T$ is a recursion safe \newDL TBox and $\A$ is $k$-bounded for $\T$.
For any  two CQs $q_1$, $q_2$:
		\begin{itemize}
		 \item[(\textbf{g})]  $q_1 \gsArrowT{\K}{g} q_2$ implies  $cert(q_1, \T,\A) \subseteq cert(q_2, \T, \A)$, and 
		 \item[(\textbf{s})] $q_1 \gsArrowT{\K}{s} q_2$  implies  $cert(q_2, \T,\A) \subseteq cert(q_1, \T, \A)$.
		 \end{itemize}
\end{proposition}
\begin{proof}
Let us start with statement (\textbf{g}). This clearly holds for every $q_1,q_2$, such that $q_2$ is obtained from $q_1$ using one of the rules 
\ref{gd:3}-\ref{gd:6}.  It is also straightforward that the containment holds also when generalizing using rule \ref{gd:1}.
The most interesting case is when $q_2$ is obtained from $q_1$ by means of rule \ref{gd:2}:
	let $\vec{t} \in cert(q_1,\T,\A)$ and $q^* = q_1 \cap q_2$ (hence $q^* = (q_1 \setminus \{x=a\})$). Then, there exists a match $\pi$ such that $\pi (x)=a$ and
	$(\T,\A) \ent \pi(q_1(\vec{t}))$. It must be the case that $\pi$ is a match also for $q^*$. 
	Since $(\T,\A) \ent r(a,b)$,  we can construct $\pi' = \pi \cup \{y = b\}$, which is clearly a match for $q^* \cup \{r(x,y), y=b\}$, where $y$ does not occur in $q_1$. Hence,  $\vec{t} \in  cert(q_2,\T,\A)$.
For statement (\textbf{s}), if $q_2$ if obtained from $q_1$ via any of the rules \ref{sd:3}-\ref{sd:6}, then clearly $cert(q_2,\T,\A) \subseteq cert(q_1,\T,\A)$. Since 
rules \ref{sd:1}, \ref{sd:2} imply adding more query atoms, then also in this case the proposition is straightforward.
\end{proof}


\section{CRIs for Modeling Dimensional Data} 
\label{sec:MMD} 
%

A range of applications that need to access data from multiple
perspectives and at various granularity levels adopt the so-called
\emph{multi-dimensional data model} \cite{HurtadoMendelzonMD}.  
This model is usually formalized as a set of \emph{dimensions}, comprising 
 a finite 
set of \emph{categories} and a partial order between them, 
sometimes called \emph{child-parent relation}. 
We may also have a \emph{dimension instance} that defines \emph{members} for
each category, and a child-parent relation between members of connected
categories. In Figure \ref{ex:DI} a dimension schema and instance of some $\sf Location$ hierarchy are ilustrated, which makes use of  concepts 
from ontology $\T_e$ as categories.

In what follows, we argue that the language of recursion safe \newDL together with $k$-bounded datasets, are well-suited to provide a similar
multi-dimensional description while mentaining efficient query answering over the dimensions. Moreover, the relaxing or restraining operators can be used for navigating along various granularity levels encoded by some given dimension.


\subsection{Dimensions as Order Constraints} 

\begin{algorithm}[t]
	\caption{CheckAdmissibility}
	\label{alg:Admissib}
	\SetKwInput{KwInput}{Input} \SetKwInput{KwOutput}{Output} \SetKwInput{KwReturn}{return} 
	\KwInput{$(\T, \A)$ satisfiable recursion safe \newDL KB, $\C$ - order constraints;}
	\KwOutput{$\bf true$ if $(\A,\T)$ is $\C$-admissible, $\bf false$ otherwise;}
	\ForEach{$ord(s, {\bf A}, \prec) \in \C$}{
		$q_1(x,y) \leftarrow s(x,y)$, \quad 
		$q_2(x,y) \leftarrow \underset{A_1 \prec A_2}{\bigvee} A_1(x), s(x,y), A_2(y)$ \;
		\lIf{$ans(q_1, \Smod{\T}{\A}) \not \subseteq ans(q_2,\Smod{\T}{\A})$}{	\Return $\bf false$ }	
		$q_3(x,y) \leftarrow \underset{A_1 \not\prec A_2}{\bigvee} A_1(x), s(x,y), A_2(y)$ \;
		\lIf{$ans(q_1,\Smod{\T}{\A}) \cap  ans(q_3, \Smod{\T}{\A}) \neq \emptyset$}{	\Return $\bf false$ }
	} 
	\Return $\bf true$.
\end{algorithm}

We introduce \emph{order constraints} to encode dimensions schemes as follows:
\begin{definition} 
	An \emph{order constraint} takes the form 
	$\ord(s,\bfA,\prec)$,  
	with $s \in \nrsimple$, $\bfA \subseteq \nc$ finite, and 
	$\prec$ a strict partial order over $\bfA$. 
 $\I$ satisfies  $\ord(s,\bfA,\prec)$ if \\
\begin{tabular}{p{6cm}p{6cm}}
	{\begin{align}
	&s^\I  \subseteq \underset{A_1, A_2 \in \bfA}{\bigcup}(A_1^\I \times
        A_2^\I),   \label{eqord:1} 
	\end{align}}
 &
{\begin{align}
	&s^\I \cap \underset{A_1 \nprec A_2}{\bigcup}(A_1^\I \times A_2^\I) = \emptyset. \label{eqord:2}
	\end{align}}
\end{tabular} 
\end{definition} 
Intuitively, if $\ord(s,\bfA,\prec)$ is satisfied in $\I$, then
all objects connected via role $s$ are instances of $\bfA$-concepts,
in a way that is compliant with the order $\prec$.

\begin{example} \label{examp:Dim} 
     The {\sf Location} dimension in our example 
  is captured in $\K_e' = (\T_e' , \A_e')$ by adding the constraint 
\[\mathrm{c} = \ord(\role{locatedIn}, \{\conc{Venue}, \conc{City},
  \conc{Country}\}, \prec)\] 
	where the order is 
	$\conc{Venue} \prec \conc{City} \prec \conc{Country}$, $\T_e'$ is as in Example \ref{ex:CRITBox} and $\A_e'$ extends $\A_e$ with the following assertions:
	\[ \begin{array}{ll}
	\conc{Venue}({\sf VolksTheater}), &  \role{locatedIn}({\sf VolksTheater}, {\sf Vienna}),  \\
	\conc{Venue}({\sf GarnierOpera}), &  \role{locatedIn}({\sf GarnierOpera}, {\sf Paris}),  \\ 	  
      \conc{City}({\sf Paris}),  & \role{locatedIn}({\sf Paris}, {\sf France}), \\
    \conc{Country}({\sf France}). &
	\end{array}
	\] 
In the models of $\K_e'$ that satisfy $\mathrm{c}$, the role
$\role{locatedIn}$ can only relate instances of $\conc{Venue}$ with instances
of $\conc{City}$ or $\conc{Country}$, and instances of $\conc{City}$ with only instances
of $\conc{Country}$. This holds, in particular, for $\Smod{\T_e'}{\A_e'}$, as
well as for the universal model $\I^{\T_e',\A_e'}$. 
\demo
\end{example}

An useful insight is that order constraints can provide $k$-bounded guarantees.

\begin{definition} 
Let $\T$ be recursion-safe. 
We say that 
 $\C$  \emph{covers $\T$} if for each $S_r = \{s \mid r \cdot s \ISA r \in \T \}$ there exists a strict partial order $(\prec, \bfA)$ such that for each 
 $s \in S_r$ there is $ \ord(s,\bfA',\prec) \in \C$, where $\bfA' \subseteq \bfA$. We
       call $(\T,\A)$ \emph{$\C$-admissible} if 
 $\Smod{\T}{\A} \ent \C$.
\end{definition}
%

\begin{example} \label{examp:order} 
The set  $\{\rm c\}$ with ${\rm c}$ the order constraint from the previous
example
covers $\T_e'$, and since $\Smod{\T_e'}{\A_e'} \ent \{\rm c\}$, we have that $(\T_e', \A_e')$  is
$\{\rm c\}$-admissible. 
 	\demo
\end{example}

\noindent
$\C$-admissibility guarantees $k$-boundedness, for $k$ determined by 
constraints in  $\C$.

\begin{lemma} 
	Let $(\T,\A)$ be a recursion-safe \newDL KB, and let $\C$ be a
        set of order constraints that covers $\T$. 
   Let $\ell(\C) =max \{ |\bfA| \mid ord(s, \bfA, \prec) \in \C \}$.
 If $(\T,\A)$ is $\C$-admissible, then $\A$ is $\ell(\C)$-bounded for $\T$.
	\label{lemma:adm-kbound}
\end{lemma}

\begin{proof}[Proof (sketch)]
For any $\I$, if $\I \models \ord(s,\bfA,\prec)$, for each chain 
	of individuals $a_1, \dots, a_{n}$ with $(a_i,a_{i+1}) \in
        s^\I$ for all $1 \leq j < n$, we have $n \leq |\bfA|$.  
        This applies to $\Smod{\T}{\A}$, as $(\T,\A)$ is $\C$-admissible. 
    Further, $\C$ covers $\T$, so for each $S_r = \{s \mid r \cdot s \ISA r \in \T\}$,
       all $S_r$-paths in $\Smod{\T}{\A}$ have size $\leq \ell(\C)$.
        Finally, all $S_r$-paths in $\A$ w.r.t.\,$\T$ are also in
        $\Smod{\T}{\A}$, so their length is $\leq \ell(\C)$. 
\end{proof}

Lemmas  \ref{claim:kboundFO} 
and \ref{lemma:adm-kbound} give us the desired result: we obtain
FO-rewritability in the presence of CRIs, 
whenever order constraints allow us to guarantee boundedness. 

\begin{theorem}
	Let $\T$ be a  recursion safe \newDL TBox, $\C$ a set of order
        constraints that covers $\T$, and  $q$ a CQ.  
Let $q_{\C}$ be the $\ell(\C)$-rewriting of $q$ w.r.t.\, $\T$, where  $\ell(\C) =max \{ |\bfA| \mid ord(s, \bfA, \prec) \in \C \}$.
	Then, for each ABox $\A$ such that $(\T,\A)$ is consistent and
        $\C$-admissible, $cert(q,\T, \A)= cert(q_{\C},\emptyset,\A )$. 
\end{theorem}

The last ingredient we need to leverage this result is an efficient way to
test for $\C$-admissibility.  
This can be done efficiently using the procedure in
Algorithm~\ref{alg:Admissib}, which evaluates some queries over our small
model  $\Smod{\T}{\A}$, and runs in time that is polynomial in $\C$, $\T$, and
$\A$.  Note that the queries have very restricted shape, and
that all variables are mapped to ABox individuals; 
answering them is not only tractable, but likely to be efficient in 
practice. Moreover, although the test for $\C$-admissibility is data
dependent, 
it does not depend on any input query, so once it is established,
 FO-rewritability is guaranteed for any CQ. 


\begin{proposition}
	Checking $\C$-admissibility for recursion-safe \newDL KBs
       is feasible in polynomial time in combined complexity.
\end{proposition}

\subsection{Dimensional Navigation using Query Reformulation Operators}
\begin{figure}\centering
	\begin{tikzpicture}
		\draw [rounded corners] (0,5) rectangle (7,4);
		\node (Loc) [right] at (5.7,4.7) {$\rm Location$};
		\node (All) [right] at (2.5,4.5) {$\sf All$};
	
	\draw [rounded corners](0,3.5) rectangle (7,2.5);
	\node (Country) [right] at (5.75,3.2) {$\rm Country$};
	
	\node (AT) [right] at (1.3,3) {$\sf Austria$};
	\node (FR) [right] at (3.4,3) {$\sf France$};

	\draw [rounded corners](0,2) rectangle (7,1);
	\node (City) [right] at (6,1.7) {$\rm City$};
	\node (Vie) [right] at (1.3,1.5) {$\sf Vienna$};	
	\node (Par) [right] at (3.5,1.5) {$\sf Paris$};
	
	\draw [rounded corners](0,0.5) rectangle (7,-0.5);
	\node (Venue) [right] at (5.9,0.2) {$\rm Venue$};
	\node (StateOpera) [right] at (0,0) {$\sf StateOpera$};
	\node (VTh) [right] at (2,0) {$\sf VolksTheater$};
	\node (GarnierOpera) [right] at (3.7,0) {$\sf GarnierOpera$};
	
	
	\draw[->] (StateOpera)--(Vie);
	\draw[->] (VTh)--(Vie);
	\draw[->] (GarnierOpera)--(Par);
	
	\draw[->] (Vie)--(AT);
		\draw[->] (AT)--(All);
	
	\draw[->] (Par)--(FR);
	\draw[->] (FR)--(All);
	
	
	
	\draw[->, dashed]  (City)--(Country);
	\draw[->, dashed]  (Country)--(Loc);
	\draw[->, dashed]  (Venue)--(City);
	
	\end{tikzpicture}
	\caption{Dimension {\sf Location}. The categories are
Venue, City, Country, Location, and the dashed arrows represent the child-parent relation between them. 
The objects in the boxes are the members of each category, and  the solid arrows illustrate the child-parent 
relation between them.}
	\label{ex:DI}
\end{figure}
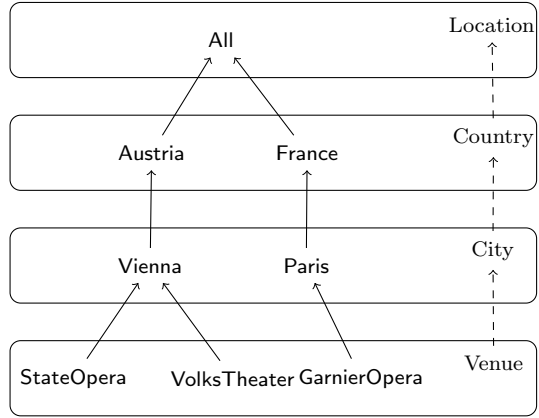
The purpose of the multi-dimensional data model is that it allows us to 
navigate data along the  different axes given by the dimensions, similarly to points in
a multi-dimensional space. 
This view lies, in fact, at the core of OLAP and similar data analytic 
applications.

In the OBDA setting, $\C$-admissible \newDL  and our query
reformulation operators enable such a 
dimensional navigation of data. 
To obtain the desired behavior, we need to ensure that the ontology contains
suitable CRIs for the roles that are used to access the dimensions. 
For example, if we use the role  $\role{occursIn}$ to query for the locations of 
events, we need to include the CRI that we have used in our examples: 
\begin{align}\role{occursIn \cdot locatedIn} \subsume
  \role{occursIn} \label{eqCRI}\end{align}
However, with these axioms and the description of the dimension explained
above, we can move up and down in the dimension retrieving data to different
levels of granularity.  
In the following examples, we illustrate how `rolling up' and `drilling down'
on the dimension are achieved.


\begin{example}[Roll up]
Let $\T_L$ be the extension of $\T_e$ with the CRI
\eqref{eqCRI}, and let $\A_L$ contain the assertions in $\A_e$ together with the assertions in
Example~\ref{examp:Dim}. 
We let $\alpha_1$ and $\alpha_2$ denote the following assertions in $\A_L$, to
which we refer below. 
\begin{align*}
\alpha_1 = & \role{locatedIn}({\sf StateOpera}, {\sf Vienna}) \\
\alpha_2 = & \role{locatedIn}({\sf Vienna}, {\sf Austria}) 
\end{align*}
Consider the following query: 
	\begin{align*}
		q(x)\leftarrow&\conc{Concert}(x), \role{occursIn}(x,y), y= {\sf Vienna}.
	\end{align*}
which returns as answers all concerts that occur  in Vienna (including, as
discussed, those that occur in a venue located in Vienna). 

We would like to `roll up' to the level of country in our location dimension,
and we expect to obtain the query
\begin{align*}
	q'(x)\leftarrow&\conc{Concert}(x), \role{occursIn}(x,y), y= {\sf Austria}.
\end{align*}
Indeed, we can obtain $q'$ from $q$ by applying  the following sequence of rules to $q$:
\begin{align*}
	  &y= {\sf Vienna} \in q, \text{ using }\alpha_2, \overset{\ref{gd:2}}{\Longrightarrow}  \\
	 &\begin{array}{ll}\quad q_1(x) \leftarrow &\conc{Concert}(x), \role{occursIn}(x,y),\role{locatedIn}(y,z), z={\sf Austria}\end{array} \\
	 &\role{occursIn}(x,y),	\role{locatedIn}(y,z) \in q_1,  \text{ using }\eqref{eqCRI}, \overset{\ref{g:6}}{\Longrightarrow}  \\
	&\qquad q_2(x)\leftarrow\conc{Concert}(x), \role{occursIn}(x,z), z = {\sf Austria}.
\end{align*}
%
The query is obtained by using relaxing rules only, that is,  $q
\gsArrowT{\K}{g}q_2$. So we know that all the answers of $q$ are preserved
when rolling up to $q'$. \demo
\end{example}
Of course, we are also interested in drilling down along our dimension in the converse
direction. 
\begin{example}[Drill down]
Starting again from the query
	\begin{align*}
		q(x)\leftarrow&\conc{Concert}(x), \role{occursIn}(x,y), y= {\sf Vienna}.
	\end{align*}
we now drill down  from the level of city to the level of venue. 
We perform the following operations on $q$:
 \begin{align*}
 &\role{occursIn}(x,y) \in q, \text{ using }\eqref{eqCRI} , \overset{{\ref{s:6}}}{\Longrightarrow}  \\
	&\begin{array}{ll}q_3(x)\leftarrow&\conc{Concert}(x), \role{occursIn}(x,z^q),\role{locatedIn}(z^q,y), y= {\sf Vienna}; \end{array}\\
  &\role{locatedIn}(z,y) \in q,  \text{ using } \alpha_1, \overset{\ref{sd:2}}{\Longrightarrow}\\
  &\begin{array}{ll}q_4(x)\leftarrow&\conc{Concert}(x), \role{occursIn}(x,z^q), z^q={\sf StateOpera},
   \role{locatedIn}(z^q,y),y= {\sf Vienna};\end{array} 
 \end{align*} 
We have
$q \gsArrowT{\K}{s} q_4$, which means that our drill down may restrict
the set of answers, but will not result in new ones.  \demo
\end{example}

\section{Related Work.}
The importance of CRIs was acknowledged since the earliest DL research, 
when \emph{role value maps} where considered very desirable \cite{Schmidt-SchaubB:1989:RVM}. 
The practical usefulness of CRIs lead to their inclusion in the OWL standard, 
both in  the OWL EL profile which is based on ${\cal EL}^{++}$ \cite{Baader2005}, 
and in full OWL 2 which is based on  ${\cal  SROIQ}$
\cite{Horrocks:2006:SROIQ}. 
Our work is also related to \emph{regular path queries} (RPQs) and their
extensions. In fact, the kind of query answering we advocate is
naturally supported in any ontology mediated setting where the DL has CRIs, or
the query language contains conjunctive RPQs; many such settings have been
considered in the literature and their complexity is well understood, see
\cite{DBLP:conf/amw/Ortiz13,DBLP:conf/rweb/OrtizS12} for references. 
However, any such combination is necessarily \textsc{NLogSpace}-hard in data complexity,
and the combined complexity is usually \textsc{PSpace}-hard even for lightweight DLs
\cite{Bienvenu:2015:RPQ}. Our focus here was on regaining FO-rewritability,
and tractable combined complexity.

 Approaches to query reformulation by removing or relaxing conditions to  return more answers is a 
problem that has been extensively studied in various communities\cite{Chaudhuri90,HuangLZ12,DologSWD09,InoueW11}.
%
%
For answering SPARQL queries over RDF data, typical relaxation steps consist in replacing a class by a superclass, a property by a superproperty, a URI or literal node by a variable, and also inserting/removing properties in property paths\cite{HurtadoPW08,ElbassuoniRW11}. 
Most of these approaches are based on similarity measures or use simple (RDF) ontologies to retrieve additional answers of possible relevance \cite{ReddyK10,HuangLZ08,HuangL10,VirgilioMT13}. 
%
The work by~\cite{DologSWD09} proposes an approach for relaxing  queries based on a rule-based query rewriting framework for RDF queries, This rather general approach is guided by domain knowledge dependent  preferences, and user preferences. \cite{FrosiniCPW17} propose query processing algorithms based on query rewriting  for SPARQL extended with query approximation and relaxation operators. 
A principled logical approach is followed by~\cite{InoueW11} for defining relaxations of conjunctive queries in so-called cooperative knowledge bases.
\cite{MartinenghiT14} propose a data model and query languages to support query relaxation over relational data. Their approach relies on simple taxonomies classifying terms used in the schema and data according to ad hoc generalization/specialization relationships. One focal point of the work is the development of abstract query languages for expressing relaxed queries over relational databases.  

The \emph{many-answers} problem, where given an initial query that returns a large number of answers has been studied for structured databases.
In this case, interactive faceted search approaches~\cite{RoyWDNM08,KashyapHP10} implementing effective drill-down strategies for helping  the user find acceptable results with minimum effort have been proposed. These approaches however, do not make use of domain knowledge. The work by~\cite{FacetedSearch2016,FacetedSearchAgg2017} address the theoretical underpinnings of faceted search in the context
of RDF and knowledge graphs. The main focus of the work on faceted search is to provide mechanisms  enabling  exploration of the underlying data and ontology, rather than on the deliberate construction of queries. An approach for evaluating  queries under generalization/specialization relations is presented in~\cite{AndreselOS16}. This work proposes a compilation technique that minimize data access in an OBDA setting, that allows to explore answers to queries along generalization/specialization steps.

In contrast to all the above work, our focus is on providing formalizations of the notions of generalizations and specialization wrt. ontological knowledge, and a principled extension for representing knowledge about multidimensional data in DLs without incurring in an increase of the data complexity of answering CQs. We consider the \dllite family of DLs as an starting point for our study, since these logics are well suited for OBDA. 
%
%

%


%
%

The notion of dimension used here is basis of the \emph{multi-dimensional
	data model} used for online-analytical processing (OLAP)
\cite{HurtadoMendelzonMD}.
Logic-based formalizations of dimensions and  multi-dimensional data schemata
have been proposed in the literature. 
Some works focus on modeling  such data and use DLs to reason about the
models, rather than for querying \cite{Franconi99adata,FranconiKamble2004}. 
A recent work in the database area focuses 
on operators for taxonomy-based relaxation of queries over 
relational data \cite{MartinenghiT14}.
Our work is closely related
to \cite{Bertossi:2018:OMD}, but  they rely on an expressive
fragment of Datalog$^{\pm}$ where dimensional knowledge can be easily
leveraged at the expense of higher complexity (i.e., not FO-rewritable).




\section{Discussion and Conclusions}
In this paper we have motivated the use of CRIs for getting more complete answers in the OBDA setting, and
we have introduced  extensions of  \dlLiteR that allow restricted forms of CRIs to preserve  FO-rewritability.  The restriction to simple roles in CRIs (Definition \ref{def:CRI}) guarantees that recursion is \emph{linear}  and  avoid a possible explosion in the size 
of rewritings.  An investigation of \newDL without
this restriction is left for future work, as well as studying CRIs and order constraints in other description logic languages.
In our first extension, \newDLnr, we disallow recursive CRIs and we showed that it is FO-rewritable, however, even for this restrictive case, CRIs lead to intractability (testing consistency is co-NP-complete). Next, we showed that recursion-safe condition allows recursive CRIs while preserving polynomial  complexity for consistency testing and instance query anwering. Lastly, we proved that if the ABox satisfies certain conditions, namely if the chains which trigger the recursive CRIs are bounded, then FO-rewritability is ensured.

We presented query reformulation rules that produce query relaxations and restrictions over any dataset, and 
more fine-grained rules that leverage the existing data. 
Finally, we have argued that \emph{admissibility of order constraints}  can describe multidimensional data, 
and that our reformulation rules
enable  navigation along dimensions, while preserving or refining answers. 
In our  query reformulation section, we have proposed a  few data-driven rules which 
intuitively take into account some patterns in the data. While there are multiple such patters that can be used for reformulating queries, our focus on these particular rules was motivated by our purpose to enhance dimensional navigation. Testing those patterns amounts to  a test for containment of certain answers for restricted queries over \newDL KBs that satisfy special properties: TBox is recurion safe and ABox is $k$-bounded. In this case, such test is $AC_0$ in data complexity. As a future research direction, we might  explore  additional data patterns that can support more flexible query reformulation. 

In  this paper, we have focused on query reformulations that either relax or restrain queries.  
It would be desirable to efficiently compute the answers to these reformulations. Thus,  we plan to investigate mechanisms  for compiling 
the data and the ontology  to support efficient answering of reformulated queries. 
Suitable  syntax and semantics of a declarative query language, in which relaxing or 
restraining operators   are first-class citizens, would definitely benefit OBDA.
Regarding the relation with multidimensional data, it would be interesting to consider aggregation and investigate whether  our operators are suitable 
for data analysis tasks, much like what OLAP systems are currently supporting.


\bibliography{bibfile}
\bibliographystyle{alpha}

\newpage

\section*{Appendix}

We start by extending the standard definition of  the \emph{chase
   procedure} of (Calvanese et al.\,2007) to \newDL. 
This construction will be used in the proofs below. 

\begin{definition}[Chase procedure]
	Let ${\cal K}=(\T, \A)$ be a \newDL KB and $f$ is a function which takes as input a set of assertions $\Gamma$ and a TBox axiom $\alpha$, and outputs the effect of applying $\alpha$ to $\Gamma$.
 We assume  the TBox to be normalized.
 We define $$chase({\cal K}) = \underset{j \in \mathbb{N}}{\bigcup}{\cal
   S}_j,$$ with ${\cal S}_0 = \A$, and ${\cal S}_{j+1}={\cal S}_j \cup
 f_{\alpha}(\Gamma)$, where  $\Gamma \subseteq {\cal S}_j$:
	\begin{itemize}
		\item if $\Gamma = \{A(a)\}$, and $A'(a) \notin {\cal S}_j$, then $f_{A \subsume A'}(\Gamma) = \{A'(a)\}$;
		\item if $\Gamma = A(a)$, and there is no individual $b$ such that $r(a,b) \in {\cal S}_j$, then $f_{A \subsume \exists r}(\Gamma)=\{r(a,a_{new})\}$;
		\item  if $\Gamma = r(a,b) $, then$f_{\exists r \subsume A}(\Gamma) = \{A'(a)\}$;
		\item if $\Gamma= \{r(a,b)\} $, and $s(a,b) \notin {\cal S}_j$, then $f_{r \subsume s}(\Gamma)=\{s(a, b)\}$;
		\item if $\Gamma = \{r(a,b)\}$, and $s(b,a) \notin {\cal S}_j$, then $f_{r \subsume s^-}(\Gamma)=\{s(b, a)\}$;
		\item if $\Gamma = \{s(a,b),t(b,c)\} $, and $r(a,c) \notin {\cal S}_j$, then $f_{s \cdot t \subsume r} (\Gamma)= \{r(a,c)\}$;
	\end{itemize}
	where $A, A' \in \nc$, $r,s,t \in \nr$, $a_{new} \in \nind \setminus \Sigma_{{\cal S}_j}$.	
\end{definition} 
We denote $chase_i(\K)=\underset{0 \leq  j \leq i}{\bigcup}{\cal S}_j$ to be the chase obtained after $i$ applications of the above rules.
We can construct from the chase, in a natural way, an interpretation $\I_{chase(\T,\A)}$, which in fact represents the \emph{canonical model} of the KB.
The following claim states that we can use $\I_{chase(\T,\A)}$, to evaluate CQs.
\begin{claim}
	Let $(\T, \A)$ be a \newDL KB and  $\I_{chase(\T,\A)}$ the interpretation constructed from $chase(\T, \A)$. If $(\T, \A)$ is consistent, then the following hold:
	\begin{enumerate}
		\item[(i)] $\I_{chase(\T,\A)}$ is a model of $(\T,\A)$; and
		\item[(ii)] for any CQ $q$: $cert(q,\T,\A)= ans(q, \I_{chase(\T,\A)})$.
	\end{enumerate}
\end{claim}

The proof of the claim is again an extension of the same proof for \dlLiteR, which is done by showing that there exists a homomorphism from $\I_{chase(\T,\A)}$ to any model of the KB.



\subsection*{Correctness of query rewriting}



\begin{lemma2}{3}\label{lemma:perfectref} 
	Let $\T$ be a \newDLnr TBox, $q$ a CQ. For every ABox $\A$ consistent with
	$\T$:
	\[cert(q,\T,\A) = \underset{q' \in \mathit{rew}(q,\T)}{\bigcup}
	cert(q',\emptyset,\A).\]
\end{lemma2}

\begin{proof}  ~\paragraph{Direction "$\supseteq$"}
	We show that for each $q' \in  \mathit{rew}(q,\T)$ we have that $cert(q', \emptyset, \A) \subseteq cert(q,\T ,\A)$. When $q'$ is obtained in one of cases $\bf S1-5$ or $\bf S7$, then it holds from (Calvanese et al, 2007). We argue for $\bf S6$: let $q'$ be obtained from $q$ by applying atom substitution $\theta = [r(x,y)/\{t(x,z^q),(z^q,y)\}]$. Then $t \cdot s \ISA r \in \T$ and $r(x,y) \in q$.
	Let $\pi$ be a match of $q'$ in $\I_\A$, which is the interpretation constructed in a natural way from ABox $\A$, such that $\pi(x)=a$, $\pi(z^q)=b$, and $\pi(y)=c$, where $a,b,c \in ind(\A)$. Then, $\pi$ is also a match of $q'$ in $\I_{chase(\T,\A)}$ the interpretation constructed from $chase(\T,\A)$, therefore there must be that $t(a,b),s(b,c) \in chase(\T,\A)$. From the construction of the chase, we get that $r(a,c) \in chase(\T,\A)$. Since $z^q \notin vars(q)$ and if $\Gamma \in q$ and $\Gamma \neq r(x,y)$, then $\Gamma \in q'$. Therefore, the match $\pi \mid_{(vars(q))}$ ($\pi$ restricted to $vars(q)$) is also a match of $q$ in $\I_{chase(\T,\A)}$. The interpretation $\I_{chase(\T,\A)}$ represents the cannonical mode of $(\T,\A)$ over which certain answers of CQs can be obtained. Therefore we can conclude $cert(q,\T,\A) \supseteq  cert(q',\emptyset,\A)$, for each $q' \in \mathit{rew}(q,\T)$.
	
	 ~\paragraph{Direction "$\subseteq$"}
	Let $\vec{t} \in cert_{\A}(\T,q)$ and we assume that $chase_k(\T,\A)$ contains a match for $q(\vec{x})$. We define ${\cal G}_k \subseteq  chase_k(\T,\A)$ to be a \emph{witness} of $\vec{t}$ w.r.t. $q$ in $chase_k(\T,\A)$, if there exists a substitution $\varphi$ from existentially quantified variables in $q$ to individuals in ${\cal G}_k$
	 such that ${\cal G}_k = \varphi(q(\vec{t}))$. For $i \in \{0,\dots,k\}$, ${\cal G}_{k-i}$ is a \emph{pre-witness} of $\vec{t}$ w.r.t. $q$ in $chase_k(\T,\A)$, and it is defined as follows: \\
	\quad $\begin{aligned}
	{\cal G}_{k-i} = \{ &\Gamma \in chase_{k-i}(\T,\A) \mid  \text{ there exists }\alpha_1, \dots, \alpha_i \\
	& \text{such that } f_{\alpha_i}(f_{\alpha_{i-1}}(\dots (f_{\alpha_1}(\Gamma)))) \in {\cal G}_k\}
	\end{aligned}$ \\
	A pre-witness contains the set of assertions in $chase_{k-i}(\T,\A)$ that trigger assertions in $ {\cal G}_k$, through application of TBox axioms.
	We have to show by induction that for each pre-witness ${\cal G}_{k-i}$ there exists $q' \in \mathit{rew}(q,\T)$ such that ${\cal G}_{k-i}$ is a witness of $\vec{t}$ w.r.t. $q'$ in $chase_{k-i}$.
	
	\textit{Base step}: $i=0$, then  $q \in \mathit{rew}(q,\T)$ and $chase_k(\T,\A)$ contains ${\cal G}_k$ which is a witness of $\vec{t}$ w.r.t. $q$.
	
	\textit{Induction step}: Assume that for ${\cal G}_{k-i+1}$ there exists $q' \in \mathit{rew}(q,\T)$ such that ${\cal G}_{k-i+1}$ is a witness of $\vec{t}$ w.r.t. $q'$ in $chase_{k-i+1}$. 
	
We do a case distinction according to the axiom used to obtain $chase_{k-i+1}$ from $chase_{k-i}$. We only give here the proof for  $\alpha= s \cdot t \subsume r \in \T$, which is new, and for $\alpha=  A \subsume  \exists r \in \T$ which requires rule {\bf S7}.
The other cases are analogous adaptations of the proof in (Calvanese et al.\,2007). 

\begin{itemize}
\item	Let $chase_{k-i+1}$ be obtained from $chase_{k-i}$ by applying axiom $\alpha= s \cdot t \subsume r \in \T$.  Let $s(a,b),t(b,c) \in chase_{k-i}$ such that $r(a,c) \notin chase_{k-i}$, then $chase_{k-i+1} = chase_{k-i} \cup \{r(a,c)\}$.
	 If $r(a,c) \notin {\cal G}_{k-i+1}$, then ${\cal G}_{k-i}={\cal G}_{k-i+1}$ therefore the claim holds. If $r(a,c) \in {\cal G}_{k-i+1}$ then ${\cal G}_{k-i}$ contains $s(a,b),t(b,c)$ and there must be some $r(x,y) \in q'$. From rule {\bf S6}, $\alpha$ is applicable to $q'$, obtaining $q'' = (q' \setminus \{r(x,y)\})\cup \{s(x,x'),t(x',y)\}$ where $x'$ is a fresh variable. Then the substitution of $q'$ in ${\cal G}_{k-i+1}$ can be extended for mapping $x'$ to $b$, hence ${\cal G}_{k-i}$ is a witness of $\vec{t}$ w.r.t. $q''$ in $chase_{k-i}$.

\item	Let  $chase_{k-i+1}$ be obtained from $chase_{k-i}$ by applying axiom $\alpha=  A \subsume  \exists r \in \T$. Let $A(a) \in chase_{k-i}$ such that $r(a, a_{new}) \notin chase_{k-i}$, then $chase_{k-i+1} = chase_{k-i} \cup \{r(a, a_{new})\}$. It follows that $r(x,y) \in q'$ and suppose $\alpha$ is not applicable to $q'$, then there must be another atom $\varPhi(y) \in q'$. Since $a_{new}$ is a fresh constant not occuring anywhere else in ${\cal G}_{k-i+1}$, it must be that $\varPhi(y)$ is mapped to 
	 $r(a, a_{new})$, hence $r(z,y) \in q'$. Using rule {\bf S7} we can substitute $z$ with $x$ and obtain query $q'' = q' \setminus \{r(z,y)\}$. Using rule {\bf S2} we obtain $\hat{q}= (q'' \setminus \{r(x,y)\}) \cup \{A(x)\}$. Hence ${\cal G}_{k-i}$ is a witness of $\vec{t}$ w.r.t $\hat{q}$ in $chase_{k-i}(\T,\A)$.	
\end{itemize}
\end{proof}

\subsection*{Consistency testing for \newDLnr}

\begin{theorem2}{1}
	Let $\varphi$ be a 3SAT formula and $(\T_{\varphi}, \A_{\varphi})$ be the \newDLnr KB reduction of $\varphi$. Then, $\varphi$ is unsatisfiable iff $(\T_{\varphi}, \A_{\varphi})$ is satisfiable. 
\end{theorem2}

\begin{proof}
	~\paragraph{Direction "if":} Assume that $\varphi$ is unsatisfiable, then for each variable assignment $\tau$, there exists some clause $c_j$ such that $\tau(c_j) = \bot$. Let $\I_{chase(\T_{\varphi},\A_{\varphi})}$ be the interpretation constructed from $chase(\T_{\varphi},\A_{\varphi})$.
	It is clear that $\I_{chase(\T_{\varphi},\A_{\varphi})}$ is a model of all positive inclusions of $\T_{\varphi}$ as well as a model of $\A_{\varphi}$. From the definition of the chase procedure, we generate new individuals for each $A_i \ISA \exists r_{x_i}$ and $A_i \ISA \exists \overline{r}_{x_i}$, therefore $\I_{chase(\T_{\varphi},\A_{\varphi})} \ent {\bf disj}(r_{x_i}, \overline{r}_{x_i})$ for each propositional variable $x_i \in \varphi$. It easily follows that in $\I_{chase(\T_{\varphi},\A_{\varphi})}$, for each instance $b$ of $A_{m+1}$  there exists a clause $c_j$ such that $b$ is not an instance of $\exists (s_{c_j}^*)^-$, meaning that the variable assignment of the path from the root individual $a$ to the leaf $b$ does not satisfy 
	clause $c_j$. Therefore, from axioms (2) and (3) it follows $\exists (p_j^{j-1})^- \ISA \exists t_j$ does not trigger on this path, therefore since $j \leq n$, and path $(a,b)$ is arbitrarily chosen we get that $\I_{chase(\T_{\varphi},\A_{\varphi})} \ent \exists t_n \ISA \bot$.
	
	~\paragraph{Direction "iff":} If $(\T_{\varphi}, \A_{\varphi})$ is satisfiable then $\I_{chase(\T_{\varphi},\A_{\varphi})}$ is a model. Since $\I_{chase(\T_{\varphi},\A_{\varphi})} \ent \exists t_n \ISA \bot$, then from axioms (2) and (3) we get that for  instance of $A_{m+1}$ is not an instance of $\exists (s_{c_1}^*)^- \sqcap \dots  \sqcap \exists (s_{c_n}^*)^-$; hence for some $A_{m+1}(b) \in chase(\T_{\varphi},\A_{\varphi})$ and some role $s_{c_j}^*$, for all $1 \leq i < m$, we get that $b$ is not an instance of $\exists (s_{c_j}^i)^-$ in $\I_{chase(\T_{\varphi},\A_{\varphi}}$. Therefore clause $c_j$ is not satisfied by the variable assignement denoted by path from root $a$ to $b$, and since $b$ an arbitrary leaf node, we get that $\varphi$ is unsatisfiable.
	

	
\end{proof}

\subsection*{Recursion safe \newDL}
\begin{proposition2}{1}
	Let $\T = \T_p \cup \T_n$ be a recursion safe \newDL TBox, where $\T_p$ contains only positive inclusions, and $\T_n$ contains only disjointness axioms. Then, for every ABox $\A$:
	\begin{enumerate}
		\item[$\bf P1$] \label{C1} If $(\T,\A)$ is satisfiable, then $\Smod{\T}{\A} \ent (\T, \A)$; 
		\item[$\bf P2$] \label{C2} $(\T,\A)$ is inconsistent iff  $\Smod{\T}{\A}
		\not\models \alpha$ for some $\alpha \in \T_n$.
		\item[$\bf P3$] \label{C3} If $\A$ is consistent with $\T$, then for any instance query $q$, we have that $cert_{\A}(q,\T) = ans(q, \Smod{\T}{\A})$.
	\end{enumerate}
\end{proposition2}

\begin{proof}[Proof] 
	We first proceed by showing a key property of $\Smod{\T}{\A}$.  For any given interpretation $\I$ and any $d,d' \in \Delta^{\I}$, let $tp_{\I} (d) = \{B \mid d \in B^{\I}\}$, and $tp_{\I}(d,d') = \{r \mid (d,d') \in r^{\I} \}$.
 	Assume that $(\T,\A)$ is consistent and let $\I$ be an arbitrary model. 
	We proceed by showing the following claim:
	\begin{claim}
		For any given $d \in \Delta^{\Smod{\T}{\A}}$ {\emph{(i)} there exists $e \in \Delta^{\I}$ such that $tp_{\Smod{\T}{\A}}(d) \subseteq tp_{\I}(e)$ and \emph{(ii)} for each $d' \in \Delta^{\Smod{\T}{\A}}$ such that $tp_{\Smod{\T}{\A}} (d,d') \neq \emptyset$ we have that there exists $e' \in \Delta^{\I}$ such that $tp_{\Smod{\T}{\A}} (d,d') \subseteq tp_{\I} (e,e')$}.
		\label{prop:1}
	\end{claim}
	
	
	\begin{itemize}
		\item We start by proving the claim for each $d \in D_0$.
		Assume there exists $B \in tp_{\Smod{\T}{\A}}(d)$ and $B \notin tp_{\I}(d^{\I})$. If $B(d) \in \A$,  this directly leads us to a contradiction since $\I$ is a model. As induction hypothesis, we assume there exists some $B' \in tp_{\Smod{\T}{\A}}(d)$ such that $B' \in tp_{\I}(d^\I)$. Then, either
		\inparaenum[a)]{\item $B' \ISA^* B \in \T$, case in which the contradiction is obvious, or \item $B= \exists r$ and here we can distinguish two sub-cases: \inparaenum[1.]{\item $B' \ISA \exists ^*r'$ and $r' \ISA^* r$ in $\T$, case that is also straightforward, or \item there exists $d_1,d_2 \in \Delta^{\Smod{\T}{\A}}$ such that either $(d,d_1) \in t^{\Smod{\T}{\A}}$, $(d_1,d_2) \in s^{\Smod{\T}{\A}}$ and $t \cdot s \ISA r \in \T$, or $(d_1,d_2) \in t^{\Smod{\T}{\A}}$, $(d_2,d) \in s^{\Smod{\T}{\A}}$ and $t \cdot s \ISA r^-$; since $s \in \nrsimple$ and $\exists s$ does not occur on rhs of any axiom of $\T$, only case 5 (in the definition of $\Smod{\T}{\A}$) can be applied for obtaining $(d_1,d_2) \in s^{\Smod{\T}{\A}}$, respectively $(d_2,d) \in s^{\Smod{\T}{\A}}$ hence $d_2 \in D_0$ and if $d_1 \in D_0$, then clearly  $\I \ent B(d)$; if $d_1 \in D_1$, then  $d_1=c_{dt}$ or $d_1= c_{d_2t^-}$ and $B^* \in tp_{\Smod{\T}{\A}}(d)$, respectively $B^* \in tp_{\Smod{\T}{\A}}(d_2)$ and $B^* \ISA t^{(-)}$. By our assumption it must be that $B^* \in tp_{\I}(d)$, or respectively $B^* \in tp_{\I}(d_2)$, hence again we obtain that $\I \ent B(d)$; the case when $d_1 \in D_2$ is not possible since, there must be some $a \in D_1$ which implies the existence of $d_1$, and since $\T$ allows only non-simple roles to generate both $a$ and $d_1$, case 6 would not be applicable.   }}
		We can follow same reasoning for proving statement \emph{(ii)} of Claim \ref{prop:1}. Let $d' \in \Delta^{\Smod{\T}{\A}}$ be such that $tp_{\Smod{\T}{\A}}(d,d') \neq \emptyset$. If $d' \in D_0$ then clearly $d' \in \Delta^{\I}$, and we easily get that $tp_{\Smod{\T}{\A}}(d,d') \subseteq tp_{\I}(d^\I, d'^\I)$. As we argued above $d' \not\in D_2$. If $d' \in D_1$, then there exists $B \ISA \exists r \in \T$ and since $B \in tp_{\I}(d)$, we get that there exists $e \in \Delta^{\I}$ such that $(d,e) \in r^{\I}$ and it follows from case b.2 above that $tp_{\Smod{\T}{\A}} (d,d') \subseteq tp_{\I} (d,e)$.
		
		\item[-] For each $d \in D_1$, there must be some $d' \in D_0$ such that $B \in tp_{\Smod{\T}{\A}(d')}$ and $B \ISA \exists r \in \T$, hence there must be some $e \in \Delta^{\I}$ such that $(d',e) \in r^{\I}$, hence $\exists r^- \in tp_{\I}(e)$. Arguing that $tp_{\Smod{\T}{\A}}(d) \subseteq tp_{\I}(e)$ is the same as above. 
		Let $d' \in \Delta^{\Smod{\T}{\A}}$ such that $tp_{\Smod{\T}{\A}} (d,d') \neq \emptyset$. Due to the restrictions of $\T$ on simple roles, it follows that for each $r \in tp_{\Smod{\T}{\A}} (d,d')$, $r$ is non-simple. The case when $d' \in D_0$ is already discussed. If $d' \in D_1$ then clearly $tp_{\Smod{\T}{\A}}(d,d') = \emptyset$ since no CRI can be applied to connect $d$ and $d'$. Lastly, if $d' \in D_2$, then there exists some $B \in tp_{\Smod{\T}{\A}}(d)$ such that $B \ISA \exists r' \in \T$, hence there must be some $e' \in \Delta^{\I}$ such that $(e,e') \in r'^{\I}$. 
		Assume there exists some $t$ such that $t \in tp_{\Smod{\T}{\A}}(d,d')$ and $t \not \in tp_{\I}(e,e')$. Since each $r,t \in tp_{\Smod{\T}{\A}}(d,d')$ and they are non-simple roles, then it must be that $r \ISA^* t \in \T$ hence a contradiction is obtained. 
		
		\item[-] The case when $d \in D_2$ is analogous to previous case. 
	\end{itemize}
	
	~\paragraph{P1:} From the definition of $\Smod{\T}{\A}$ we get that all non-disjointness axioms are satisfied, as well as $\A$. 
	Assume that there exists  $\alpha = {\bf disj}(B_1,B_2) \in \T$ (for role disjointness axioms is analogous) such that $\Smod{\T}{\A} \not \ent \alpha$. Then, there must be some $d \in \Delta^{\Smod{\T}{\A}}$ such that $B_1, B_2 \in tp_{\Smod{\T}{\A}}$, hence using Claim \ref{prop:1}, for each model $\I$ we have that $B_1,B_2 \in tp_{\I}(d)$, which leads to a contradiction with the fact that the KB is satisfiable. We can now conclude that $\Smod{\T}{\A} \ent (\T,\A)$.
	~\paragraph{P2:} If $(\T,\A)$ is inconsistent, then there exists some $\alpha \in \T_n$ such that for each $\I \ent (\T_p,\A)$ we have $\I \not \ent \alpha$. Since $\Smod{\T}{\A}$ is always a model of $(\T_p, \A)$, we get that direction "if" holds.  The other direction follows almost immediately from Claim \ref{prop:1}.
	~\paragraph{P3:}  From Proposition \ref{prop:1} follows that for any instance query and any match in $\Smod{\T}{\A}$ we can easily construct a match in any model $\I$. The other direction follows from $\bf P1$.
\end{proof}

\begin{lemma2}{4}
	Let $\T$ a recursion safe \newDL TBox, $\T_k$ be a $k$-unfolding of $\T$, for some $k \geq 0$, and $q$ a CQ over $\Sigma_{\T}$.
	Then, for every $k$-bounded ABox $\A$: 
	$$cert_(q, \T, \A) = \underset{q' \in rew(\hat{q}, \T_k) }{\bigcup}cert(q', \emptyset, \A) $$
	\label{claim:kboundFO}
\end{lemma2}

\begin{proof} 
	Let $\T^* = \T \cap \T_{k}$, then 
	we have that:
	\begin{align}
	chase(\T^*, \A)&\subseteq chase(\T, \A) \text{ and } \nonumber\\
	chase(\T^*, \A) &\subseteq chase(\T_k, \A) \label{commonchase}
	\end{align}
	
	\noindent We proceed by showing "$\supseteq$" of the lemma: 
	
	For that we need to show the following claim:
	\begin{claim}
		For each $r(a,b) \in chase(\T,\A)$ such that $r \cdot s \subsume r \in \T$, there exists $r^*(a,b) \in chase(\T_k, \A)$.
		\label{claim:rstar}
	\end{claim}
	
	\textit{Proof of Claim \ref{claim:rstar}}
	If $r(a,b)$ is produced by axioms in $\T^*$ then $r(a,b) \in chase(\T^*, \A)$, and since $r \ISA  r^* \in \T_{k}$ we get $r^*(a,b) \in chase(\T_{k},\A)$.
	If $r(a,b)$ is produced by $r\cdot s_j \subsume r \in \T$, , then using the fact that $\A$ is $k$-bounded, we distinguish the following two cases:
	\begin{enumerate}
		\item[(i)] $r(a,b) \in \A$, or
		\item[(ii)] $r(a,a'), \, s^i(a',b) \in chase{(\T^*, \A)}$, where $1 \leq i \leq k$. 
	\end{enumerate}
	For case \emph{(i)} the claim trivially holds since $ r \subsume r^*$. For case \emph{(ii)}, 
	since $r(a,a') \in chase{(\T^*, \A)}$ then $r(a,a') \in chase{(\T_k, \A)}$. Then from axioms $r \subsume r_0$, $r_{i-1} \cdot s \subsume r_{i}$, $r_i \subsume r^*$ in $\T'$, we conclude $r^*(a,b) \in \I^{(\T', \A)}$.
	\textit{End proof.}
	
	Therefore, we obtain that for each match $\pi$ of $q$ in $\I_{chase{(\T, \A)}}$,  $\pi$ is also a match for $q'$ in  $\I_{chase{(\T_k, \A)}}$. 
	
	We proceed now with "$\subseteq$" of the lemma:	Let $\pi'$ be a match for $q'$ in
	$\I_{chase{(\T_k, \A)}}$ and assume that $r^*(a,b) \in \pi'(q')$. 
	Using Lemma \ref{lemma:perfectref}, we can apply the query rewriting rules on $q'$ and $\T'$. Then, axioms
	$r_i \subsume r^*$, $r_{i-1} \cdot s \subsume r_i$, $r \subsume r_0$,
	become succesively applicable, yelding query $q_i$ containing $r(x,y),
	s^{i}(y,z)$, for each $0\leq i \leq k$. Therefore, either
	$r(a,b) \in \A$ or $r(a,c), s^{j}(c,b) \in chase{{(\T^*, \A)}}$, where
	$1 \leq j \leq k$, and since it must be that $r \cdot s \ISA r \in \T$ we can conclude that $\pi'$ is also a match for $q$ in $\I_{chase{(\T, \A)}}$.
\end{proof}



\subsection*{Order constraints and $\C$-admissibility}

We first provide a slightly modified version of the covering definition, otherwise the following claim would not hold in every case.

\begin{lemma} 
	Let $(\T,\A)$ be a recursion-safe \newDL KB, and let $\C$ be a
	set of order constraints that covers $\T$. 
	Let $\ell(\C) =max \{ |\bfA| \mid ord(s, \bfA, \prec) \in \C \}$.
	If $(\T,\A)$ is $\C$-admissible, then $\A$ is $\ell(\C)$-bounded for $\T$.
	\label{lemma:adm-kbound}
\end{lemma}

\begin{proof}
Since $\C$ covers $\T$ we get that for each $s \in S_r$, where $r$ is recursive in $\T$ there is a unique set of concept names $\bfA$ and a strict partial order $\prec$ over  $\bfA$ such that $ord(s, \bfA, \prec) \in \C$. We proceed with showing that in each $S_r$-path in $\Smod{\T}{\A}$ has size $\leq \ell(\C)$. 

Assume there exists in $\Smod{\T}{\A}$ some $S_r$-path of size $\ell(\C)+1$. Therefore, there exists $d_1, \dots, d_{\ell(\C)+1} \in ind(\A)$ such that $(d_i,d_{i+1}) \in s^{\Smod{\T}{\A}}$, where $s \in S_r$. Since $(\T,\A)$ is $\C$-admissible, then each $ord(s,\bfA , \prec) \in \C$, where $s \in S_r$, is satisfied in $\Smod{\T}{\A}$. Hence, for each $s \in S_r$ we have that $s^{\Smod{\T}{\A}} \subseteq A_1^{\Smod{\T}{\A}} \times A_2^{\Smod{\T}{\A}}$ such that $A_1 \prec A_2 \in (\bfA ,\prec)$. We can then construct a concept hierarchy $A_1 \prec A_2 \dots \prec A_j$, where $j = \ell({\C})+1$ and each $A_j \in \bfA$, hence we obtain a contradiction with the fact that $\ell(\C)$ is the size of the maximal concept order in $\C$. Therefore each $S_r$-path in $\Smod{\T}{\A}$ has size at most $\ell(\C)$.
Due to recursion-safety conditions in $\T$, and the construction of $\Smod{\T}{\A}$ we get that for each $S_r$-path in $\Smod{\T}{\A}$ there exists an $S_r$-path in $\A$, hence $\A$ is $\ell(\C)$-bounded for $\T$.
\end{proof}


\end{document}